\documentclass[a4paper, 4pt]{elsarticle}
\usepackage{lipsum}
\usepackage{arydshln}
\usepackage{amssymb}
\usepackage{graphicx}
\usepackage{amsfonts}
\usepackage{geometry}
\usepackage{epsfig}
\usepackage{multicol}
\usepackage{CJK}
\usepackage{amsmath}
\usepackage{amsthm}
\usepackage{setspace}
\usepackage{graphicx}
\usepackage{epstopdf}
\usepackage{subeqnarray}
\usepackage{mathrsfs}
\usepackage{bbding}
\usepackage{cases}
\usepackage{subfig}
\usepackage{mathrsfs}
\usepackage{float}
\usepackage{tikz}
\usepackage{geometry}
\usepackage{lineno,hyperref}
\modulolinenumbers[5]
\journal{journal}
\newcommand{\ii}{\mathrm{i}}
\newcommand{\ee}{\mathrm{e}}
\newcommand{\dd}{\mathrm{d}}

\newtheorem{prop}{Proposition}

\biboptions{sort&compress}
\begin{document}
	\title{The Riemann-Hilbert approach for the integrable fractional Fokas--Lenells equation}
	\author[]{Ling An}
	\author[]{Liming Ling\corref{mycorrespondingauthor}}
	\ead{linglm@scut.edu.cn}
	\cortext[mycorrespondingauthor]{Corresponding author}
		\address{School of Mathematics, South China University of Technology, Guangzhou, China, 510641}
	
	\begin{abstract}
		In this paper, we propose a new integrable fractional Fokas--Lenells  equation by using the completeness of the squared eigenfunctions, dispersion relation, and inverse scattering transform. To solve this equation, we employ the Riemann-Hilbert approach. Specifically, we focus on the case of the reflectionless potential with a simple pole for the zero boundary condition. And we provide the fractional $N$-soliton solution in determinant form. Additionally, we prove the fractional one-soliton solution rigorously. Notably, we demonstrate that as $|t|\to\infty$, the fractional $N$-soliton solution can be expressed as a linear combination of $N$ fractional single-soliton solutions.
	\end{abstract}

	\maketitle
	\section{Introduction}
	In the past few decades, the nonlinear Schr\"odinger (NLS) equation has been widely studied in various fields, and it can describe many nonlinear phenomena such as nonlinear optics, ion acoustic waves in plasmas, Bose-Einstein condensation, and deep-water wave propagation \cite{hasegawa1973transmission,agrawal2011nonlinear}. However, when describing the propagation of ultrashort pulses in nonlinear optical fibers, higher-order effects become crucial, necessitating the appropriate introduction of some higher-order terms \cite{potasek1991exact,cavalcanti1991modulation,kodama1997input}. Each of these terms has a specific contribution to the pulse propagation, and the choice of coefficients for these higher-order terms often determines whether the model is integrable or not. In a similar way to the perturbation of the bi-Hamiltonian structure of the famous Korteweg-de Vries (KdV) equation to obtain the integrable Camassa-Holm equation \cite{fuchssteiner1981symplectic}, the application of the same mathematical technique to the two associated Hamiltonian operators of the NLS equation results in the Fokas--Lenells (FL) equation \cite{fokas-1995class,lenells-2009exactly,lenells-2008novel,lenells-2010dressing}, which is an important generalization of the NLS equation, 
\begin{equation}\label{2-FL1}
	\ii q_{t}-\alpha_{1}q_{tx}+\alpha_{2}q_{xx}+\sigma|q|^{2}(q+\ii\alpha_{1}q_{x})=0,\ \ \ \sigma=\pm 1,\ \ \alpha_{1},\alpha_{2}\in\mathbb{R}.
\end{equation}
If we do the transformation $q\to\beta_{2}\sqrt{\beta_{1}}\ee^{\ii\beta_{2}x}q,\ \sigma\to-\sigma,\ \beta_{1}=\frac{\alpha_{2}}{\alpha_{1}},\  \beta_{2}=\frac{1}{\alpha_{1}}$, then we can obtain
\begin{equation}\label{2-FL2}
	q_{tx}+\beta_{1}\beta_{2}^{2}q-2\ii\beta_{1}\beta_{2}q_{x}-\beta_{1}q_{xx}+\ii\sigma\beta_{1}\beta_{2}^{2}|q|^{2}q_{x}=0.
\end{equation}
We can also do another transformation $	q\to\gamma_{2}\sqrt{\gamma_{1}}\ee^{\ii(\gamma_{2}x+2\gamma_{1}\gamma_{2}t)}q,\  \gamma_{1}=\frac{\alpha_{2}}{\alpha_{1}}>0,\ \gamma_{2}=\frac{1}{\alpha_{1}},\ \chi=x+\gamma_{1}t,\ \tau=-\gamma_{1}\gamma_{2}^{2}t$, which can yield
\begin{equation}\label{2-FL3}
	q_{\tau\chi}-q+\ii\sigma|q|^{2}q_{\chi}=0.
\end{equation}
An important feature of the equation \eqref{2-FL3} is that it describes the first negative flow of the integrable hierarchy associated with the derivative nonlinear 
Schr\"odinger (DNLS) equation, known as the Kaup-Newell (KN) system. The FL equation, in contrast to the NLS equation, exhibits a notable distinction in that the latter possesses both a focusing and defocusing version depending on parameter values. However, all variants of the FL equation \eqref{2-FL1}-\eqref{2-FL3} are mathematically identical, with differences only arising from variable transformations.

The FL equation, which possesses the Lax pair and integrability, has garnered significant attention from researchers, resulting in a variety of interesting findings. In \cite{vekslerchik2011lattice}, the authors studied the lattice representation of the FL equation and constructed corresponding dark soliton solutions. The bright/dark soliton solutions of the FL equation were systematically obtained under vanishing/nonvanishing boundary conditions in \cite{matsuno2012direct-1,matsuno2012direct-2}. The rogue wave solutions of the FL equation were investigated in detail in \cite{he-2012rogue}. In \cite{zhao2013algebro}, the authors presented the functional representation of algebro-geometric solutions to the FL equation and reduced them to $n$-dark soliton solutions via the degeneration of related Riemann surfaces. Three types of Darboux transformations for the FL equation were constructed, leading to the multi-soliton solutions in \cite{wang2020fokas}. The bilinear approach was utilized to investigate the FL equation in \cite{liu2022fokas}. Additionally, the integrable nonlocal form of the FL equation and its exact solutions were presented in \cite{zhang2019exact,li2021n}. Numerous studies on the FL equation exist, and we will not enumerate them all here. The aforementioned research on the FL equation and its nonlocal form demonstrate that they both belong to integrable differential equations of integer order. However, when describing the dynamics of certain systems, fractional differential equations can better reflect the actual variation law of the system \cite{wang-2020dynamical}. This naturally prompts us to search for the integrable fractional form of the FL equation.
	
In $2022$, Ablowitz, Been and Carr introduced a fractional NLS equation and a fractional KdV equation \cite{ablowitz2022fractional}, where the fractional operators were defined using Riesz fractional derivative \cite{Riesz1949,agrawal2001applications,lischke2020fractional}. The remarkable feature of these fractional differential equations is their integrability in the sense of inverse scattering transform (IST). Subsequently, Ablowitz et al. extended these fractional differential equations to the general Ablowitz-Kaup-Newell-Segur
	(AKNS) system \cite{ablowitz2022integrable} and investigated the fractional discrete NLS equation \cite{ablowitz2022discrete}. In \cite{weng2022dynamics,zhang2022interaction,yan2022multi}, the authors explored the fractional forms of the higher-order NLS equation, higher-order modified KdV equation, and given a new integrable multi-L$\acute{\rm e}$vy-index and mixed fractional nonlinear equations. Additionally, they employed the deep learning method to analyze the fractional equations \cite{zhong2022data}. In \cite{an2022nondegenerate}, the authors analyzed a fractional coupled Hirota equation, which extended the $2\times 2$ fractional Lax equation to the $3\times 3$ case. The fractional DNLS equation was investigated in \cite{an2023inverse}, expanding the integrable fractional AKNS system to the KN system. Furthermore, in \cite{mou-2023integrable}, the authors presented the fractional form of the $n$-component coupled NLS equation. The above research has provided us an effective idea for exploring the fractional form of FL equation. And we will discuss the fractional form of the following FL equation:
\begin{equation}\label{2-FL}
	q_{xt}-q+\ii\sigma|q|^{2}q_{x}=0.
\end{equation}
In the field of physics, the IST plays a crucial role in studying nonlinear wave equations with Lax pairs. In solving nonlinear integrable models, an advanced approach called the Riemann-Hilbert method, which is an extension of the IST, has gained widespread utilization \cite{biondini2014inverse,biondini-2016inverse,guo-2012}. Consequently, we employ the Riemann-Hilbert approach in our search for exact solutions to the fractional FL (fFL) equation.

This paper is organized as follows. In Sec.$2$, we extend the recursion operator function $\mathcal{F}(\mathcal{L})=\mathcal{L}^{-1}$, which corresponds to the FL equation, to its fractional form $\mathcal{F}_{fF}(\mathcal{L})=\mathcal{L}^{-1}|2\mathcal{L}|^{2\epsilon}$. This extension allows us to obtain the operator representation for the fFL equation. Additionally, we establish an appropriate Riemann-Hilbert problem and provide its solution. Moreover, we explore the relations between perturbations in the potential function and the scattering data, deriving the completeness of the corresponding squared eigenfunctions. By combining this result with the operator representation for the fFL equation, we derive the exact form of the fFL equation. In Sec.$3$, we present a determinant representation for the fractional $N$-soliton solution. We rigorously prove the fractional one-soliton solution and determine the fractional rational solution by considering the limit form of the fractional one-soliton solution. Based on the exact expressions of the fractional one-soliton and rational solutions, we analyze the effect of the fractional small parameter $\epsilon$ on these solutions. Furthermore, we investigate the decomposed property of the fractional $N$-soliton solution as $|t|$ tends to $\infty$.

\section{IST for the fFL equation}
In this section, we want to construct the exact form of the fFL equation and solve it by IST. Since the FL equation and the DNLS equation all belong to the KN system, we can combine the idea in \cite{an2023inverse} to find the integrable fractional form of the FL equation.

Firstly, we consider the following spectral problem:
\begin{subequations}\label{2-Lax-Phi}
	\begin{equation}\label{2-Lax-Phi-a}
		{\bf \Phi}_{x}={\bf U}{\bf \Phi},\ \ \ \ {\bf U}(\lambda;x,t)=-\ii\lambda^{2}\sigma_{3}+\lambda{\bf \widehat{Q}}(x,t),\ \ \ {\bf \widehat{Q}}(x,t)=\begin{bmatrix}
			0 & \hat{q}\\[2pt]
		\hat{r} & 0
		\end{bmatrix},
	\end{equation}
	\begin{equation}\label{2-Lax-Phi-b}
		{\bf \Phi}_{t}={\bf V}{\bf \Phi},\ \ \ \ {\bf V}(\lambda;x,t)=\begin{bmatrix}
			V_{1}&V_{2}\\[4pt]
			V_{3}&-V_{1}
		\end{bmatrix},
	\end{equation}
\end{subequations}
where $\lambda\in\mathbb{C}\cup\{\infty\}$ is the spectral parameter, ${\bf \Phi}={\bf \Phi}(\lambda;x,t)$ is the wave function, $\hat{q}=\hat{q}(x,t)$ and $\hat{r}=\hat{r}(x,t)$ are the potential functions, $V_{j}=V_{j}(\lambda;x,t),\ j=1,2,3$ are the polynomial functions, $\sigma_{3}$ is the Pauli's spin matrix, and here we directly introduce three Pauli's spin matrices:
\begin{equation*}
	\sigma_{1}=\begin{bmatrix}
		0&1\\
		1&0
	\end{bmatrix},\ \ \ \ 
	\sigma_{2}=\begin{bmatrix}
	0&-\ii\\
	\ii&0
\end{bmatrix},\ \ \ \ 
	\sigma_{3}=\begin{bmatrix}
	1&0\\
	0&-1
\end{bmatrix}.
\end{equation*}
To ensure the compatibility of \eqref{2-Lax-Phi}, the zero curvature equation need to be satisfied,
	\begin{equation*}
	{\bf U}_{t}-{\bf V}_{x}+[{\bf U},{\bf V}]=0,\ \ \ \ [{\bf U},{\bf V}]\equiv {\bf UV-VU},
\end{equation*}
which implies
\begin{equation}\label{2-V}
	\begin{split}
		V_{1x}-\lambda(\hat{q}V_{3}-\hat{r} V_{2})=0,&\\
		V_{2x}+2\ii\lambda^{2}V_{2}+2\lambda \hat{q}V_{1}-\lambda \hat{q}_{t}=0,&\\
		V_{3x}-2\ii\lambda^{2}V_{3}-2\lambda \hat{r}V_{1}-\lambda \hat{r}_{t}=0.&
	\end{split}
\end{equation}
The equation \eqref{2-V} can be rewritten as
\begin{equation}\label{2-V-rewrit}
	\lambda\begin{bmatrix}
		\hat{q}\\[2pt]\hat{r}
	\end{bmatrix}_{t}=2\lambda V_{10}\sigma_{3}\begin{bmatrix}
	\hat{q}\\[2pt]\hat{r}
\end{bmatrix}+(\mathcal{L}_{1}+2\ii \lambda^{2}\mathcal{L}_{2})\sigma_{3}\begin{bmatrix}
V_{2}\\[2pt]V_{3}
\end{bmatrix},
\end{equation}
where $V_{10}$ is an integration constant, and
\begin{equation*}
	\begin{split}
	\mathcal{L}_{1}=\sigma_{3}\partial,\ \ \ \ \mathcal{L}_{2}=\mathbb{I}+\ii\sigma_{3}\begin{bmatrix}
   \hat{q}\\[2pt]\hat{r}
   \end{bmatrix}\partial^{-1}_{-}\begin{bmatrix}
   	\hat{r},&\hat{q}
   \end{bmatrix},\ \ \ \ \partial=\partial/\partial_{x},\ \ \ \ \partial^{-1}_{-}=\int_{-\infty}^{x}\dd y.
	\end{split}
\end{equation*}
Choosing $V_{10}=-\frac{\ii}{4}(-\lambda^{2})^{n}$,
\begin{equation*}
	\begin{bmatrix}
		V_{2}\\[2pt]V_{3}
	\end{bmatrix}=\sum_{j=1}^{n}(-1)^{n-j}\begin{bmatrix}
	V_{2j}\\[2pt]V_{3j}
\end{bmatrix}\lambda^{2(n-j)+1},
\end{equation*}
substituting these expansions into \eqref{2-V-rewrit}, and collecting the same powers of $\lambda$, then we can get the hierarchy:
\begin{equation}\label{2-evolution eq}
	\begin{bmatrix}
	\hat{q}\\[2pt]\hat{r}
	\end{bmatrix}_{t}=-\frac{\ii}{2}\mathcal{F}(\mathcal{L})\begin{bmatrix}
	\hat{q}\\[2pt]-\hat{r}
\end{bmatrix},\ \ \ \ \mathcal{F}(\mathcal{L})=\mathcal{L}^{n},
\end{equation}
where
\begin{equation*}
	\mathcal{L}=\frac{1}{2\ii}\mathcal{L}_{1}\mathcal{L}_{2}^{-1}=-\frac{1}{2}\left(\ii\sigma_{3}\partial+{\bf u}_{x}\partial^{-1}_{-}{\bf u}^{\top}\sigma_{1}+{\bf u}{\bf u}^{\top}\sigma_{1}\right)=-\frac{1}{2}\begin{bmatrix}
		\ii\partial+\hat{q}_{x}\partial^{-1}_{-}\hat{r}+\hat{q}\hat{r}&\hat{q}_{x}\partial^{-1}_{-}\hat{q}+\hat{q}^{2}\\[4pt]
		\hat{r}_{x}\partial^{-1}_{-}\hat{r}+\hat{r}^{2}&-\ii\partial+\hat{r}_{x}\partial^{-1}_{-}\hat{q}+\hat{q}\hat{r}
	\end{bmatrix},
\end{equation*}
${\bf u}(x,t)=\begin{bmatrix}
	\hat{q},&\hat{r}
\end{bmatrix}^{\top}$, the superscript $^{\top}$ denotes the transpose. $\mathcal{L}$ is the adjoint of
\begin{equation*}
\widetilde{\mathcal{L}}=-\frac{1}{2}\begin{bmatrix}
	-\ii\partial-\hat{r}_{x}\partial^{-1}_{+}\hat{q}+\hat{q}\hat{r}&\hat{r}_{x}\partial^{-1}_{+}\hat{r}-\hat{r}^{2}\\[4pt]
	\hat{q}_{x}\partial^{-1}_{+}\hat{q}-\hat{q}^{2}&\ii\partial-\hat{q}_{x}\partial^{-1}_{+}\hat{r}+\hat{q}\hat{r}
\end{bmatrix},\ \ \ \ \partial^{-1}_{+}=\int_{x}^{+\infty}\dd y.
\end{equation*}
The hierarchy \eqref{2-evolution eq} can yield many integrable equations by choosing different values of $n$. For example, the DNLS equation can be obtained by choosing $n=2$ and $\hat{r}=\sigma \hat{q}^{*}\ (\sigma=\pm 1)$, where the superscript $^{*}$ denotes the complex conjugate. In addition, if we take $n=-1$, $\hat{r}=-\sigma\hat{q}^{*}\ (\sigma=\pm 1)$, and make a transformation $\hat{q}=q_{x}$, then we can derive the FL equation \eqref{2-FL}.

The operator function $\mathcal{F}(\mathcal{L})$ can be related to
the dispersion relation of the linearization of \eqref{2-evolution eq}, and  $\mathcal{F}(\mathcal{L})$ can be generalized to the more general form $\mathcal{F}(\lambda)$ by using the properties of the squared eigenfunctions which will be discussed later. The linearization of \eqref{2-evolution eq} with respect to $\hat{q}$ is as follows:
\begin{equation}\label{2-evolution-q-linear}
\hat{q}_{t}(x,t)=-\frac{\ii}{2}\left(-\frac{\ii}{2}\partial\right)^{n}\hat{q}(x,t).
\end{equation}
Then we put $\hat{q}(x,t)\thicksim\ee^{\ii(\lambda^{2} x-\omega(\lambda^{2})t)}$ into \eqref{2-evolution-q-linear}, which can yield
\begin{equation}\label{2-relation-F-ome}
	\mathcal{F}\left(\frac{\lambda^{2}}{2}\right)=2\omega\left(\lambda^{2}\right).
\end{equation}
Obviously, the relation \eqref{2-relation-F-ome} will be invariant after the transformation $\hat{q}=q_{x}$, then we can directly obtain $\omega_{F}(\lambda^{2})=\lambda^{-2}$ which corresponds to the FL equation \eqref{2-FL}. Following the rule in \cite{ablowitz2022fractional}, we assume the dispersion relation of the fFL equation is $\omega_{fF}(\lambda^{2})=\lambda^{-2}|\lambda^{2}|^{2\epsilon},\ \epsilon\in[0,1)$, then the linearization of the fFL equation is
\begin{equation*}
	 q_{xt}-|-\partial^{2}|^{\epsilon}q=0,
\end{equation*}
where $|-\partial^{2}|^{\epsilon}$ is called the Riesz fractional derivative \cite{ablowitz2022fractional}. Combining the relation \eqref{2-relation-F-ome}, there is
\begin{equation*}
	\mathcal{F}_{fF}(\lambda)=\lambda^{-1}|2\lambda|^{2\epsilon},
\end{equation*}
which will lead to the operator function $\mathcal{F}_{fF}(\mathcal{L})=\mathcal{L}^{-1}|2\mathcal{L}|^{2\epsilon}$ of the fFL equation.

\subsection{Direct scattering}\label{subsec-diret}
We rewrite the Lax pair of the fFL equation as
\begin{equation}\label{2-Lax pair-fFL}
	{\bf\Phi}_{x}=(-\ii\lambda^{2}\sigma_{3}+\lambda {\bf Q}_{x}){\bf\Phi},\ \ \ \ {\bf\Phi}_{t}={\bf V}{\bf\Phi},\ \ \ \ {\bf Q}=\begin{bmatrix}
		0&q\\[2pt]
		r&0
	\end{bmatrix},
\end{equation}
where $r=-\sigma q^{*},\ \sigma=\pm 1$, and without loss of generality, we take $\sigma=-1$.
Assume that the potential function $q(x,t)$ is sufficiently smooth and rapidly tends to zero as $|x|\to\infty$. Note that the matrix ${\bf V}(\lambda;x,t)$ cannot be given explicitly for the fFL equation, while the constraint ${\bf V}(\lambda;x,t)\to\frac{\ii}{4}\mathcal{F}_{fF}(\lambda^{2})\sigma_{3},\ |x|\to\infty$ need to be satisfied. 

In Sec.\ref{subsec-diret}, we will perform 
spectral analysis on the $x$-part of \eqref{2-Lax pair-fFL}, thus the variable $t$ temporarily be considered as a dummy variable. From the matrix spectral problem \eqref{2-Lax pair-fFL}, we have the following asymptotic behavior:
	\begin{equation}\label{2-asy-Phi}
		{\bf\Phi}^{\pm}=
		\begin{bmatrix}
			\phi^{\pm}_{1},&\phi^{\pm}_{2}
		\end{bmatrix}\to
		\exp\left(-\ii\lambda^{2}\sigma_{3}x\right),\ \ \ x\to\pm\infty,
	\end{equation}
where the superscripts $^{\pm}$ refer to the cases of $x\to\pm\infty$, respectively. 

This motivates us to introduce the variable transformation:
\begin{equation}\label{2-psi-phi}
	{\bf\Phi}={\bf\Psi}\exp\left(-\ii\lambda^{2}\sigma_{3}x\right),
\end{equation}
which can yield the canonical normalization:
\begin{equation}\label{2-asy-Psi}
	{\bf\Psi}^{\pm}=\begin{bmatrix}
		\psi^{\pm}_{1},&\psi^{\pm}_{2}
	\end{bmatrix}\to\mathbb{I},\ \ \ \  x\to\pm\infty,
\end{equation}
where the superscripts $^{\pm}$ also refer to the cases of $x\to\pm\infty$, respectively. 
Inserting \eqref{2-psi-phi} into the first equation in \eqref{2-Lax pair-fFL}, there is
\begin{equation}\label{2-lax-Psi}
	{\bf\Psi}_{x}=-\ii\lambda^{2}[\sigma_{3},{\bf\Psi}]+\lambda{\bf Q}_{x}{\bf\Psi}.
\end{equation}
Then the Volterra integral equations for ${\bf\Psi}^{\pm}(\lambda;x,t)$ can be obtained,
\begin{equation}\label{2-Psi-Volterra}
	{\bf\Psi}^{\pm}(\lambda;x,t)=\mathbb{I}+\lambda\int_{\pm\infty}^{x}\ee^{-\ii\lambda^{2}(x-y){\rm ad}\widehat{\sigma}_{3}}\left({\bf Q}_{y}(y,t){\bf\Psi}^{\pm}(\lambda;y,t)\right)dy,
\end{equation}
where $\ee^{{\rm ad}\widehat{\sigma}_{3}}{\bf X}=\ee^{\sigma_{3}}{\bf X}\ee^{-\sigma_{3}}$ with ${\bf X}$ being a $2\times 2$ matrix. In addition, the large-$\lambda$ expansions of ${\bf\Psi}^{\pm}(\lambda;x,t)$ are given by
\begin{equation*}
	{\bf \Psi}^{\pm}(\lambda;x,t)=\exp\left(\frac{\ii\sigma_{3}}{2}\int_{\pm\infty}^{x}|q_{y}(y,t)|^{2}dy\right)+\mathcal{O}(\lambda^{-1}).
\end{equation*}
Through a transformation, we can transform the linear equation \eqref{2-lax-Psi} to a spectral problem of the Zakharov–Shabat (ZS) type, thereby rewriting the Volterra integral equations for ${\bf \Psi}(\lambda;x,t)$ into the form corresponding to the ZS system. This makes it easy for us to analyze the analyticity of ${\bf \Psi}(\lambda;x,t)$, and the similar results have been reported in Lemma $1$ in \cite{pelinovsky-2018}. We summarize the analyticity of ${\bf \Psi}(\lambda;x,t)$ in the proposition below.
\begin{prop} \label{2-prop-phi-analytic}
	We assume $\partial_{x} q(x,t)\in\mathrm{L}^{1}(\mathbb{R})\cap\mathrm{L}^{3}(\mathbb{R})$ and $\partial_{xx}q(x,t)\in\mathrm{L}^{1}(\mathbb{R})$, then there exist the unique
	solutions satisfying the Volterra integral equations \eqref{2-Psi-Volterra} for every $\lambda\in\Sigma$, where $\Sigma=\mathbb{R}\cup\ii\mathbb{R}$, and ${\bf\Psi}^{\pm}(\lambda;x,t)$ have the following properties:
	\begin{itemize}
		\item  The column vectors $\psi^{-}_{1}$ and $\psi^{+}_{2}$ are analytic for $\lambda\in \mathbb{D_{+}}$ and continuous for $\lambda\in \mathbb{D_{+}}\cup\Sigma$,
		\item  The column vectors $\psi^{+}_{1}$ and $\psi^{-}_{2}$ are analytic for $\lambda\in \mathbb{D_{-}}$ and continuous for $\lambda\in \mathbb{D_{-}}\cup\Sigma$,
	\end{itemize}
where
\begin{equation*}
	\mathbb{D_{+}}=\left\{\lambda\big|\arg\lambda\in\left(0,\frac{\pi}{2}\right)\cup\left(\pi,\frac{3\pi}{2}\right)\right\},\ \ \ \ \mathbb{D_{-}}=\left\{\lambda\big|\arg\lambda\in\left(\frac{\pi}{2},\pi\right)\cup\left(\frac{3\pi}{2},2\pi\right)\right\}.
\end{equation*}
\end{prop}

We know that ${\bf\Phi}^{\pm}(\lambda;x,t)$ can be regarded as the fundamental solutions of \eqref{2-Lax-Phi}, then they are all linearly dependent on each other by the theory of ordinary differential equations,
	\begin{equation}\label{2-Phi-S}
	{\bf\Phi}^{-}={\bf\Phi}^{+}{\bf S}(\lambda;t),\ \ \ \ \lambda\in\Sigma,
\end{equation}
where ${\bf S}(\lambda;t)=\left(s_{ij}(\lambda;t)\right)_{i,j=1,2}$ is called the scattering matrix. Combining with the relation \eqref{2-psi-phi}, we have
\begin{equation}\label{2-Psi-S}
	{\bf\Psi}^{-}={\bf\Psi}^{+}\ee^{-\ii\lambda^{2}x{\rm ad}\widehat{\sigma}_{3}}{\bf S}(\lambda;t),\ \ \ \ \lambda\in\Sigma.
\end{equation}
The scattering matrix ${\bf S}(\lambda;t)$ can be expanded at infinity of $\lambda$,
\begin{equation*}
	{\bf S}(\lambda;t)=\exp\left(\frac{\ii\sigma_{3}}{2}\int_{-\infty}^{+\infty}|q_{x}(x,t)|^{2}dx\right)+\mathcal{O}(\lambda^{-1}).
\end{equation*}
By applying Abel's formula and ${\rm tr}{\bf Q}_{x}=0$, we have $\det{\bf \Psi}=1$, then $\det{\bf S}=1$ can be derived from \eqref{2-Psi-S}. Furthermore, \begin{equation*}
	\begin{split}
		&s_{11}(\lambda;t)=\left|\psi_{1}^{-},\psi_{2}^{+}\right|,\ \ \ \ s_{12}(\lambda;t)=\left|\psi_{2}^{-},\psi_{2}^{+}\right|\ee^{2\ii\lambda^{2}x},\\
		&s_{22}(\lambda;t)=\left|\psi_{1}^{+},\psi_{2}^{-}\right|,\ \ \ \ s_{21}(\lambda;t)=\left|\psi_{1}^{+},\psi_{1}^{-}\right|\ee^{-2\ii\lambda^{2}x}.
	\end{split}
\end{equation*}
	Therefore, $s_{11}(\lambda;t)$ and $s_{22}(\lambda;t)$ are analytic in $\mathbb{D}_{\pm}$, respectively. As usual, the off-diagonal scattering
coefficients cannot be extended off the contour $\Sigma$. 
	\begin{prop}\label{2-prop-sym}
	The solution ${\bf\Psi}(\lambda;x,t)$ and the scattering matrix ${\bf S}(\lambda;t)$ all have two symmetry reductions:
	\begin{itemize}
		\item ${\bf\Psi}(\lambda;x,t){\bf\Psi}^{\dagger}(-\lambda^{*};x,t)=\mathbb{I}$,\ \ \ ${\bf S}(\lambda;t){\bf S}^{\dagger}(-\lambda^{*};t)=\mathbb{I}$.
		\item ${\bf\Psi}(\lambda;x,t)=\sigma_{3}{\bf\Psi}(-\lambda;x,t)\sigma_{3}$,\ \ \ ${\bf S}(\lambda;t)=\sigma_{3}{\bf S}(-\lambda;t)\sigma_{3}$.
	\end{itemize}
\end{prop}
\begin{proof}
	Based on the equation \eqref{2-lax-Psi}, we can derive its adjoint equation:
	\begin{equation}\label{2-Psi-adjoint}
		\widetilde{{\bf\Psi}}_{x}=-\ii\lambda^{2}[\sigma_{3},\widetilde{{\bf\Psi}}]-\lambda\widetilde{{\bf\Psi}}{\bf Q}_{x}.
	\end{equation}
	Now we assume ${\bf\Psi}(\lambda;x,t)$ be a solution of
	the equation \eqref{2-lax-Psi} corresponding to the eigenvalue $\lambda$. Through simple calculations, it can be
	found that $\left({\bf\Psi}(\lambda;x,t)\right)^{-1}$ satisfies the adjoint equation \eqref{2-Psi-adjoint}, so we call $\left({\bf\Psi}(\lambda;x,t)\right)^{-1}$ a matrix
	adjoint solution with respect to $\lambda$. In addition, it is easy to find the symmetry reduction of the matrix function ${\bf Q}(x,t)$: ${\bf Q}^{\dagger}={\bf Q}$, then we have
	\begin{equation*}
		\begin{split}
			\left({\bf\Psi}^{\dagger}(-\lambda^{*};x,t)\right)_{x}=&-\ii\lambda^{2}\left[\sigma_{3},\Psi^{\dagger}(-\lambda^{*};x,t)\right]-\lambda\Psi^{\dagger}(-\lambda^{*};x,t){\bf Q}_{x}^{\dagger}(x,t)\\
			=&-\ii\lambda^{2}\left[\sigma_{3},\Psi^{\dagger}(-\lambda^{*};x,t)\right]-\lambda\Psi^{\dagger}(-\lambda^{*};x,t){\bf Q}_{x}(x,t).
		\end{split}
	\end{equation*}
	So ${\bf\Psi}^{\dagger}(-\lambda^{*};x,t)$ satisfies the equation \eqref{2-Psi-adjoint}. Applying the canonical asymptotic condition \eqref{2-asy-Psi}, $\left({\bf\Psi}(\lambda;x,t)\right)^{-1}$ and ${\bf\Psi}^{\dagger}(-\lambda^{*};x,t)$ have the same
	asymptotic behavior. Therefore, we can obtain $\left({\bf\Psi}(\lambda;x,t)\right)^{-1}={\bf\Psi}^{\dagger}(-\lambda^{*};x,t)$ due to the
	solution of the scattering problem \eqref{2-Psi-adjoint} is uniquely determined by its boundary condition. By direct calculation, we can find that $\sigma_{3}{\bf\Psi}(-\lambda;x,t)\sigma_{3}$ satisfies the equation \eqref{2-lax-Psi} and its boundary condition is consistent with ${\bf\Psi}(\lambda;x,t)$, which proves the second symmetry of ${\bf\Psi}(\lambda;x,t)$. Then the symmetry reductions of ${\bf S}(\lambda;t)$ can be obtained by combining the relation \eqref{2-Psi-S}.
\end{proof}
Based on the proposition \ref{2-prop-sym}, we have $s_{11}(\lambda;t)=s_{11}(-\lambda;t)=s_{22}^{*}(\lambda^{*};t)=s_{22}^{*}(-\lambda^{*};t)$. So the zeros of $s_{11}(\lambda;t)$ appear in pairs, and we can assume that $s_{11}(\lambda;t)$ has simple zeros defined by $\lambda_{n},\ n=1,2,\cdots,N$ in the ${\rm\uppercase\expandafter{\romannumeral1}}$ quadrant, and $\lambda_{n+N}=-\lambda_{n}$ in the ${\rm\uppercase\expandafter{\romannumeral3}}$ quadrant. That is to say, $s_{11}(\lambda_{j};t)=0,\ \frac{\partial s_{11}(\lambda;t)}{\partial\lambda}|_{\lambda=\lambda_{j}}\neq0,\ j=1,2,\cdots,2N$. Then there is $s_{22}(\lambda_{j}^{*};t)=0,\ \frac{\partial s_{22}(\lambda;t)}{\partial\lambda}|_{\lambda=\lambda_{j}^{*}}\neq0,\ j=1,2,\cdots,2N$. We call the zeros $\lambda=\lambda_{j},\ \lambda=\lambda_{j}^{*}$ of $s_{11}(\lambda)$ and $s_{22}(\lambda)$ the discrete spectrums of \eqref{2-lax-Psi}, respectively. And the equation \eqref{2-lax-Psi} also has continuous spectrums, which all lie on $\Sigma$. Here we define the discrete
spectrums by the set $\Lambda=\Lambda_{1}\cup\Lambda_{2}$,
	\begin{equation*}
		\Lambda_{1}=\big\{\lambda_{n},\ -\lambda_{n}\big\}_{n=1}^{N},\ \ \ \ \Lambda_{2}=\big\{ \lambda_{n}^{*},\ -\lambda_{n}^{*}\big\}_{n=1}^{N},
	\end{equation*}
	whose distributions are shown in Fig.\ref{fig:contour}.
	\begin{figure}
		\centering
		\includegraphics[width=0.45\linewidth]{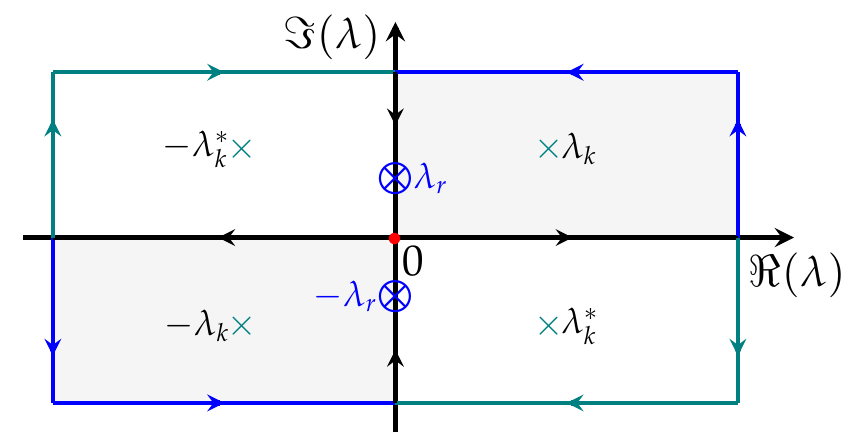}
		\caption{The complex $\lambda$-plane. Region $\mathbb{D_{+}}$ (grey region), Region $\mathbb{D_{-}}$ (white region), discrete spectra, and orientation of the
			contours for the matrix
			Riemann-Hilbert problem in the
			IST. The discrete spectra corresponding to the fractional rational solution are indicated by $\bigotimes$. The solid black lines along the grey and white areas are denoted as $\Sigma_{+}$ and $\Sigma_{-}$, respectively. And the solid blue line and the solid bluish-green line correspond to the integral paths $\Gamma_{+}$ and $\Gamma_{-}$, respectively.}
		\label{fig:contour}
	\end{figure}

\subsection{Time evolution}	
	The asymptotic behavior of the
	associated time evolution operator ${\bf V}(\lambda;x,t)$, which cannot be represented generally, is as follows:
	\begin{equation*}
		{\bf V}\to\frac{\ii}{4}\mathcal{F}_{fF}(\lambda^{2})\sigma_{3},\ \ \ \ |x|\to\infty.
	\end{equation*}
	Then the time evolution of the scattering data can be obtained,
	\begin{equation*}
		\begin{split}
			&s_{11}(\lambda;t)=s_{11}(\lambda;0),\ \ \ \  s_{12}(\lambda;t)=\ee^{-\frac{\ii}{2}\mathcal{F}_{fF}(\lambda^{2})t}s_{12}(\lambda;0),\\
			&s_{22}(\lambda;t)=s_{22}(\lambda;0),\ \ \ \ s_{21}(\lambda;t)=\ee^{\frac{\ii}{2}\mathcal{F}_{fF}(\lambda^{2})t}s_{21}(\lambda;0).
		\end{split}
	\end{equation*}

\subsection{Inverse scattering}\label{2-sec-IST}
We define a matrix function ${\bf T}(\lambda;x,t)$:
\begin{equation*}
	\begin{split}
		&{\bf T}_{+}=\begin{bmatrix}
			\psi_{1}^{-},&\psi_{2}^{+}
		\end{bmatrix}={\bf \Psi}^{-}{\bf H}_{1}+{\bf\Psi}^{+}{\bf H}_{2},\ \ \ \ \ \ \ \ \ \ \lambda\in\mathbb{D_{+}},\\
		&{\bf T}_{-}=\begin{bmatrix}
			\widetilde{\psi}^{-,1},&
			\widetilde{\psi}^{+,2}
		\end{bmatrix}^{\top}={\bf H}_{1}\widetilde{{\bf \Psi}}^{-}+{\bf H}_{2}\widetilde{{\bf\Psi}}^{+},\ \ \ \ \lambda\in\mathbb{D_{-}},
	\end{split}
\end{equation*}
where ${\bf H}_{1}={\rm diag}\left(1,0\right),\ {\bf H}_{2}={\rm diag}\left(0,1\right)$, and $\widetilde{\psi}^{\pm,j}(\lambda;x,t)$ refer to the $j$-th row of $\widetilde{\bf\Psi}^{\pm}(\lambda;x,t)$. Apparently, ${\bf T}_{\pm}(\lambda;x,t)$ are analytic in $\mathbb{D}_{\pm}$ by reviewing the analyticity of ${\bf\Psi}^{\pm}(\lambda;x,t)$. And we summarize some properties of ${\bf T}_{\pm}(\lambda;x,t)$ in the following
proposition.
\begin{prop}\label{2-prop-T}
	The matrix function ${\bf T}(\lambda;x,t)$ has the following properties:
	\begin{enumerate}
	\item[1] $\det {\bf T}_{+}(\lambda;x,t)=s_{11}(\lambda;t),\ \ \det {\bf T}_{-}(\lambda;x,t)=s_{22}(\lambda;t)$.\\
	\item[2] $\lim\limits_{x\to+\infty}{\bf T}_{+}(\lambda;x,t)=\begin{bmatrix}
		s_{11}(\lambda;t)&0\\
		0&1
	\end{bmatrix}$,\ \ $\lim\limits_{x\to-\infty}{\bf T}_{-}(\lambda;x,t)=\begin{bmatrix}
	1&0\\
	0&s_{22}(\lambda;t)
\end{bmatrix}$.\\
\item[3] ${\bf T}_{+}(\lambda;x,t)\to{\bf E}(x,t),\ {\bf T}_{-}(\lambda;x,t)\to{\bf E}^{-1}(x,t)$, {\rm as} $\lambda\to\infty$, where
\begin{equation*}
	{\bf E}(x,t)=\begin{bmatrix}
		\exp\left(\frac{\ii}{2}\int_{-\infty}^{x}|q_{y}(y,t)|^{2}\dd y\right)
		&0\\
		0&\exp\left(\frac{\ii}{2}\int_{x}^{+\infty}|q_{y}(y,t)|^{2}\dd y\right)
	\end{bmatrix}.
\end{equation*}
\item[4] ${\bf T}_{+}^{\dagger}(-\lambda^{*};x,t)={\bf T}_{-}(\lambda;x,t),\ \ {\bf T}_{\pm}(\lambda;x,t)=\sigma_{3}{\bf T}_{\pm}(-\lambda;x,t)\sigma_{3}$.
\end{enumerate}
\end{prop}
	Note that the inverse of the solution of the adjoint matrix spectral problem \eqref{2-Psi-adjoint} is exactly the solution of the matrix spectral problem \eqref{2-lax-Psi}, then $({\bf T}_{-}(\lambda;x,t))^{-1}$ is the solution of the equation \eqref{2-lax-Psi}. So we can define a sectional analytic solution ${\bf M}(\lambda;x,t)$ of the matrix spectral problem \eqref{2-lax-Psi},
		\begin{equation*}
		\begin{split}
			&{\bf M}_{+}(\lambda;x,t)={\bf T}_{+}(\lambda;x,t),\ \ \ \ \ \ \ \  \lambda\in\mathbb{D_{+}},\\
			&{\bf M}_{-}(\lambda;x,t)=({\bf T}_{-}(\lambda;x,t))^{-1},\ \ \  \lambda\in\mathbb{D_{-}}.
		\end{split}
	\end{equation*}
Based on the symmetry reductions of ${\bf T}(\lambda;x,t)$ in proposition \ref{2-prop-T}, we have ${\bf M}_{-}(\lambda;x,t){\bf M}_{+}^{\dagger}(-\lambda^{*};x,t)=\mathbb{I}$ and ${\bf M}_{\pm}(\lambda;x,t)=\sigma_{3}{\bf M}_{\pm}(-\lambda;x,t)\sigma_{3}$. 
	Then the related Riemann-Hilbert problem can be constructed.\\
	{\bf Riemann-Hilbert problem $1$.\ }We can construct a matrix function ${\bf M}(\lambda;x,t)$ with the following properties:
\begin{itemize}
	\item ${\bf Analyticity:}$\ ${\bf M}_{+}(\lambda;x,t)$ is analytic in  $\mathbb{D}_{+}$, while ${\bf M}_{-}(\lambda;x,t)$ is analytic in  $\mathbb{D}_{-}\setminus\Lambda_{2}$.
	\item ${\bf Residue\ condition:}$\ The conditions to the principal parts of the meromorphic matrix ${\bf M}_{-}(\lambda;x,t)$
	at $\lambda=\lambda_{j}^{*},\ j=1,\cdots,2N,$ are given by: 
	\begin{equation*}
		{\bf ker}\left(\underset{\lambda=\lambda_{j}^{*}}{\rm{Res}}\ {\bf M}_{-}(\lambda;x,t)\right)={\bf Im}\left({\bf T}_{-}(\lambda_{j}^{*};x,t)\right),\ \ \ \ {\bf Im}\left(\underset{\lambda=\lambda_{j}^{*}}{\rm{Res}}\ {\bf M}_{-}(\lambda;x,t)\right)={\bf ker}\left({\bf T}_{-}(\lambda_{j}^{*};x,t)\right).
	\end{equation*}
		\item ${\bf Jump\ condition:}$\ ${\bf M}(\lambda;x,t)$ satisfies the jump condition:
	\begin{equation*}
		{\bf M}_{+}(\lambda;x,t)={\bf M}_{-}(\lambda;x,t){\bf J}(\lambda;x,t),\ \ \lambda\in\Sigma,
	\end{equation*}
	where
	\begin{equation*}
		{\bf J}(\lambda;x,t)=\exp\left(\left(-\ii\lambda^{2}x+\frac{\ii}{4}\mathcal{F}_{fF}(\lambda^{2})t\right){\rm ad}\widehat{\sigma}_{3}\right)\begin{bmatrix}
			1&-s_{12}(\lambda;t)\\[4pt]
			s_{21}(\lambda;t)&1
		\end{bmatrix}.
	\end{equation*}
	\item ${\bf Asymptotic\ behavior:}$\
	\begin{equation*}
		{\bf M}(\lambda;x,t)={\bf E}(x,t)+\mathcal{O}(\lambda^{-1}),\ \ \ \ {\rm as}\ \lambda\to\infty.
	\end{equation*}
\end{itemize}
	The Riemann-Hilbert problem $1$ was named as the non-regular Riemann-Hilbert problem because of the determinants of the functions ${\bf M}_{+}(\lambda;x,t)$ and $\left({\bf M}_{-}(\lambda;x,t)\right)^{-1}$ can be zero at certain discrete locations. And the frequent application of this one is to construct the exact soliton solutions using the dressing operators \cite{wang-2010,wu-2017,ma-2019}. The uniqueness of the Riemann-Hilbert problem $1$ is not guaranteed unless the zeros of $\det {\bf T}_{\pm}(\lambda;x,t)$ in $\mathbb{D}_{\pm}$ are specified, and the kernel structures of ${\bf T}_{\pm}(\lambda;x,t)$ at their zeros are determined. Note that we have supposed that $s_{11}(\lambda;t)$ and $s_{22}(\lambda;t)$ all have $2N$ simple zeros in Sec.\ref{subsec-diret}, then the kernel of ${\bf T}_{+}(\lambda;x,t)$ contains a single column vector $|{\bf v}_{j}\rangle=\begin{bmatrix}
		v_{1j},&v_{2j}
	\end{bmatrix}^{\top}$ and the kernel of ${\bf T}_{-}(\lambda;x,t)$ contains a single row vector $\langle{\bf v}_{j}|=|{\bf v}_{j}\rangle^{\dagger}=\begin{bmatrix}
		v_{1j}^{*},&v_{2j}^{*}
	\end{bmatrix}$, i.e.,
	\begin{equation}\label{2-T-kernel}
		{\bf T}_{+}(\lambda_{j};x,t)|{\bf v}_{j}\rangle=0,\ \ \ \ \langle{\bf v}_{j}|{\bf T}_{-}(\lambda_{j}^{*};x,t)=0,\ \ \ \ j=1,2,\cdots,2N.
	\end{equation}
Here and in what follows, the ``ket" $|{\bf v}_{j}\rangle$ represents a usual column vector, and the
``bra" $\langle{\bf v}_{j}|$ is its conjugate transpose vector.

Next, we will eliminate the zero structures of ${\bf T}_{\pm}(\lambda;x,t)$ at $\lambda=\lambda_{j},\ \lambda=\lambda^{*}_{j}$ to ensure that the Riemann-Hilbert problem $1$ can be solved via Plemelj's formulae. Based on the symmetry reductions of ${\bf\Psi}(\lambda;x,t)$ in proposition \ref{2-prop-sym}, we introduce a matrix function ${\bf G}_{1}(\lambda;x,t)$ and its inverse \cite{guo-2012}:
\begin{equation*}
	{\bf G}_{1}=\mathbb{I}+\frac{{\bf A}_{1}}{\lambda-\lambda_{1}^{*}}-\frac{\sigma_{3}{\bf A}_{1}\sigma_{3}}{\lambda+\lambda_{1}^{*}},\ \ \ \ 
	{\bf G}_{1}^{-1}=\mathbb{I}+\frac{{\bf A}_{1}^{\dagger}}{\lambda-\lambda_{1}}-\frac{\sigma_{3}{\bf A}_{1}^{\dagger}\sigma_{3}}{\lambda+\lambda_{1}}.
\end{equation*}
The matrix function ${\bf A}_{1}(\lambda;x,t)$ can be derived by combining $\underset{\lambda=\lambda_{1}}{\rm{Res}}\left({\bf G}_{1}(\lambda;x,t)\left({\bf G}_{1}(\lambda;x,t)\right)^{-1}\right)=0$,
\begin{equation*}
	{\bf A}_{1}=\begin{bmatrix}
		a_{1}&0\\
		0&-a_{1}^{*}
	\end{bmatrix}|{\bf v}_{1}\rangle\langle {\bf v}_{1}|,\ \ \ \ a_{1}=\frac{\lambda_{1}^{2}-\lambda_{1}^{*2}}{(\lambda_{1}-\lambda_{1}^{*})\langle {\bf v}_{1}|\sigma_{3}|{\bf v}_{1}\rangle-(\lambda_{1}+\lambda_{1}^{*})\langle {\bf v}_{1}|{\bf v}_{1}\rangle}.
\end{equation*}
Moreover, through calculations, we can find $\det {\bf G}_{1}=\frac{\lambda^{2}-\lambda_{1}^{2}}{\lambda^{2}-\lambda_{1}^{*2}}$, and
\begin{equation*}
	{\bf G}_{1}(\lambda_{1};x,t)|{\bf v}_{1}\rangle=0,\ \ \ \ \langle {\bf v}_{1}|({\bf G}_{1}(\lambda_{1}^{*};x,t))^{-1}=0.
\end{equation*}
Then we define a new matrix function ${\bf R}_{1}(\lambda;x,t)$:
\begin{equation*}
	{\bf R}_{1,+}={\bf T}_{+}{\bf G}_{1}^{-1},\ \ \ \ {\bf R}_{1,-}={\bf G}_{1}{\bf T}_{-}.
\end{equation*}
\begin{prop}
	The matrix function ${\bf R}_{1}(\lambda;x,t)$ has the following properties:
	\begin{enumerate}
		\item[1] ${\bf R}_{1,\pm}(\lambda;x,t)$ no longer have the zero structures at $\lambda=\pm\lambda_{1}$ and $\lambda=\pm\lambda_{1}^{*}$, respectively.\\
		\item[2] $\underset{\lambda=\pm\lambda_{1}}{\rm{Res}}\ {\bf R}_{1,+}(\lambda;x,t)=0,\ \ \ \ \underset{\lambda=\pm\lambda_{1}^{*}}{\rm{Res}}\ {\bf R}_{1,-}(\lambda;x,t)=0.$\\
		\item[3] $\det{\bf R}_{1,+}(\pm\lambda_{1};x,t)\neq 0,\ \ \ \ \det{\bf R}_{1,-}(\pm\lambda_{1}^{*};x,t)\neq 0$.\\
		\item[4] ${\bf R}_{1,+}(\lambda_{2};x,t)|{\bf w}_{2}\rangle=0,\ \ \langle{\bf w}_{2}|{\bf R}_{1,-}(\lambda_{2}^{*};x,t)=0$,\ where $|{\bf w}_{2}\rangle={\bf G}_{1}(\lambda_{2};x,t)|{\bf v}_{2}\rangle$.
	\end{enumerate}
\end{prop}
\noindent
Following this rule, we can construct the matrix functions ${\bf R}_{N,\pm}(\lambda;x,t)$ which have no zero structures at $\lambda=\pm\lambda_{j}$ and $\lambda=\pm\lambda_{j}^{*}$,\ $j=1,\cdots,N$,
\begin{equation*}
	{\bf R}_{N,+}={\bf T}_{+}{\bf G}^{-1},\ \ \ \ {\bf R}_{N,-}={\bf G}{\bf T}_{-},
\end{equation*}
where ${\bf G}={\bf G}_{N}\cdots{\bf G}_{1}$,
\begin{equation*}
	\begin{split}
	&{\bf G}_{j}=\mathbb{I}+\frac{{\bf A}_{j}}{\lambda-\lambda_{j}^{*}}-\frac{\sigma_{3}{\bf A}_{j}\sigma_{3}}{\lambda+\lambda_{j}^{*}},\ \ \ \ 
	{\bf G}_{j}^{-1}=\mathbb{I}+\frac{{\bf A}_{j}^{\dagger}}{\lambda-\lambda_{j}}-\frac{\sigma_{3}{\bf A}_{j}^{\dagger}\sigma_{3}}{\lambda+\lambda_{j}},\ \ \ \ j=1,\cdots,N,\\
	&{\bf A}_{j}=\begin{bmatrix}
		a_{j}&0\\
		0&-a_{j}^{*}
	\end{bmatrix}|{\bf w}_{j}\rangle\langle {\bf w}_{j}|,\ \ \ \ a_{j}=\frac{\lambda_{j}^{2}-\lambda_{j}^{*2}}{(\lambda_{j}-\lambda_{j}^{*})\langle {\bf w}_{j}|\sigma_{3}|{\bf w}_{j}\rangle-(\lambda_{j}+\lambda_{j}^{*})\langle {\bf w}_{j}|{\bf w}_{j}\rangle},\\[3pt]
&|{\bf w}_{j}\rangle={\bf G}_{j-1}(\lambda_{j};x,t)\cdots{\bf G}_{1}(\lambda_{j};x,t)|{\bf v}_{j}\rangle.
\end{split}
\end{equation*}
Note that $\det {\bf G}_{j}=\frac{\lambda^{2}-\lambda_{j}^{2}}{\lambda^{2}-\lambda_{j}^{*2}}$, and
\begin{equation*}
	{\bf G}_{j}(\lambda_{j};x,t)|{\bf w}_{j}\rangle=0,\ \ \ \ \langle {\bf w}_{j}|({\bf G}_{j}(\lambda_{j}^{*};x,t))^{-1}=0.
\end{equation*}

Now we define a new sectional analytic function ${\bf M}^{[1]}(\lambda;x,t)$:
\begin{equation*}
	\begin{split}
	&{\bf M}^{[1]}_{+}={\bf R}_{N,+}={\bf M}_{+}{\bf G}^{-1},\ \ \ \ \ \ \ \ \  \lambda\in\mathbb{D_{+}},\\
	&{\bf M}^{[1]}_{-}=({\bf R}_{N,-})^{-1}={\bf M}_{-}{\bf G}^{-1},\ \ \ \ \lambda\in\mathbb{D_{-}}.
\end{split}
\end{equation*}
Then the new Riemann-Hilbert problem can be constructed.\\
{\bf Riemann-Hilbert problem $2$.\ }We can find a matrix function ${\bf M}^{[1]}(\lambda;x,t)$ with the following properties:
\begin{itemize}
	\item ${\bf Analyticity:}$\ ${\bf M}^{[1]}_{\pm}(\lambda;x,t)$ are analytic in $\mathbb{D}_{\pm}$, respectively.
	\item ${\bf Jump\ condition:}$\ ${\bf M}^{[1]}(\lambda;x,t)$ satisfies the jump condition:
	\begin{equation*}
		{\bf M}_{+}^{[1]}(\lambda;x,t)={\bf M}_{-}^{[1]}(\lambda;x,t){\bf J}^{[1]}(\lambda;x,t),\ \ \lambda\in\Sigma,
	\end{equation*}
	where ${\bf J}^{[1]}={\bf GJG^{-1}}$.
	\item ${\bf Asymptotic\ behavior:}$\
	\begin{equation*}
		{\bf M}^{[1]}(\lambda;x,t)={\bf E}(x,t)+\mathcal{O}(\lambda^{-1}),\ \ \ \ {\rm as}\ \lambda\to\infty.
	\end{equation*}
\end{itemize}
Obviously, the determinant of ${\bf M}^{[1]}(\lambda;x,t)$ will not equal to zero. So we can apply Plemelj's formulae to solve the above Riemann-Hilbert problem.  
	Here we introduce the Cauchy projectors $\mathcal{P}_{\pm}$ \cite{biondini2014inverse} over $\Sigma$,
\begin{equation*}
	\mathcal{P}_{\pm}[f](\lambda)=\frac{1}{2\pi\ii}\int_{\Sigma_{\pm}}\dfrac{f(\zeta)}{\zeta-(\lambda\pm\ii 0)}\dd\zeta,
\end{equation*}
where $\Sigma_{\pm}$ refer to the integral paths along the gray area and the white area indicated by arrows in Fig.\ref{fig:contour}, and the notations $\lambda\pm\ii0$ indicate that the limits are taken from the left or right of $\lambda$. Applying the Cauchy projectors $\mathcal{P}_{\pm}$ to the Riemann-Hilbert problem $2$, we can get
\begin{equation*}
	{\bf M}^{[1]}_{\pm}(\lambda;x,t)={\bf E}(x,t)+\frac{1}{2\pi\ii}\int_{\Sigma_{\pm}}\dfrac{{\bf M}_{-}(\zeta;x,t)\left({\bf J}(\zeta;x,t)-\mathbb{I}\right)\left({\bf G}(\zeta;x,t)\right)^{-1}}{\zeta-(\lambda\pm\ii 0)}\dd\zeta.
\end{equation*}

For the integrable fractional DNLS equation, its potential function is recovered by the expansion of the analytic function as $\lambda\to+\infty$ \cite{an2023inverse}. However, it is more convenient to recover the potential function as $\lambda\to 0$. Therefore, we will recover $q(x,t)$ by the expansion of ${\bf M}_{+}(\lambda;x,t)$ as $\lambda\to 0$. We expand ${\bf M}_{+}(\lambda;x,t)$ at $\lambda\to 0$ as
\begin{equation}\label{2-expand-M}
	{\bf M}_{+}(\lambda;x,t)={\bf M}_{+,0}(x,t)+\lambda{\bf M}_{+,1}(x,t)+o(\lambda),
\end{equation}
 substitute \eqref{2-expand-M} into
\eqref{2-lax-Psi}, and collect the same powers of $\lambda$, then ${\bf M}_{+,0}(x,t)=\mathbb{I},\ {\bf Q}_{x}(x,t)=\left({\bf M}_{+,1}(x,t)\right)_{x}$, i.e.,
\begin{equation*}
	q(x,t)=\left({\bf M}_{+,1}(x,t)\right)_{12}.
\end{equation*}
Consequently, we have recovered the potential function $q(x,t)$.

\subsection{Exact form of the fFL equation}	
Next, our objective is to derive the exact form of the fFL equation. According to the equation \eqref{2-evolution eq}, we need to analyze the operator function $\mathcal{F}_{fF}(\mathcal{L})$ acting on the vector $\widehat{\bf u}_{x}=\begin{bmatrix}
	q_{x},&-q^{*}_{x}
\end{bmatrix}^{\top}$. Note that the squared eigenfunctions are eigenfunctions of the recursion operator of integrable equations. Motivated by this, our approach entails transforming the problem of the recursion operator 
$\mathcal{L}$ acting on the vector $\widehat{\bf u}_{x}(x,t)$ into another problem of the recursion operator $\mathcal{L}$ acting on the squared eigenfunctions. And then we can generalize it to the case of the operator function $\mathcal{F}_{fF}(\mathcal{L})$.

The completeness relation of the squared eigenfunctions plays a crucial role in determining the exact form of the fFL equation. There exist various methods for deriving this completeness relation, each with its own advantages. 
In our study, we have adopted the method proposed by Kaup et al. in their research on the massive Thirring model and the Benjamin-Ono equation \cite{kaup-1996squared,kaup-1999perturbation}, as well as the method presented by Kaup and Yang for a third-order scattering operator \cite{yang-2009squared}. Although alternative methods, such as the one outlined in \cite{kawata1980generalized,kawata1980linear}, can also be used to derive the exact form of the fFL equation, we have found the aforementioned methods to be more convenient. Our derivation procedure begins with determining the variations of the scattering data by perturbing the potential. The resulting coefficients are then regarded as the adjoint squared eigenfunctions. Subsequently, we utilize the Riemann-Hilbert approach to compute the variation of the potential based on the variations of the scattering data, enabling us to directly obtain the squared eigenfunctions. By combining these two procedures, we can efficiently derive the completeness relation of the squared eigenfunctions.

Firstly, we calculate the variations of the scattering data by considering the variation of the potential. This allows us to determine the adjoint squared eigenfunctions. To accomplish this, we introduce a variation to the equation \eqref{2-Lax pair-fFL} and solve the resulting variation equation. Let us take the variation to the equation \eqref{2-Lax pair-fFL}, then
\begin{equation}\label{2-delta-Phi}
	\delta{\bf\Phi}_{x}=-\ii\lambda^{2}\sigma_{3}\delta{\bf\Phi}+\lambda{\bf Q}_{x}\delta{\bf\Phi}+\lambda\delta{\bf Q}_{x}{\bf\Phi},
\end{equation}
where the notation $\delta\alpha$ is a quantity related to $\alpha$ to represent $\|\delta\alpha\|$ is sufficiently small. Note that the function ${\bf\Phi}(\lambda;x,t)\to\exp\left(\left(-\ii\lambda^{2}x+\frac{\ii}{4}\mathcal{F}_{fF}(\lambda^{2})t\right)\sigma_{3}\right)$ as $|x|\to\infty$, so $\delta{\bf\Phi}(\lambda;x,t)$ tends to zero under this boundary condition. Then the variation equation \eqref{2-delta-Phi} can be solved by
\begin{equation*}
	\delta{\bf\Phi}^{\pm}(\lambda;x,t)=\lambda{\bf\Phi}^{\pm}(\lambda;x,t)\int_{\pm\infty}^{x}\left({\bf\Phi}^{\pm}(\lambda;y,t)\right)^{-1}\delta{\bf Q}_{y}(y,t){\bf\Phi}^{\pm}(\lambda;y,t)\dd y.
\end{equation*}
We will consider the above equation from the cases of $x\to\pm\infty$ by using the relation \eqref{2-Phi-S}. In the limit of $x\to+\infty$, there are $\delta{\bf\Phi}^{+}\to0,\ \delta{\bf\Phi}^{-}\to{\bf\Phi}^{+}\delta{\bf S}$, then we can get
\begin{equation}\label{2-delta-S}
	\delta{\bf S}(\lambda;t)=\lambda\int_{-\infty}^{+\infty}\left({\bf\Phi}^{+}(\lambda;x,t)\right)^{-1}\delta{\bf Q}_{x}(x,t){\bf\Phi}^{-}(\lambda;x,t)\dd x.
\end{equation}
Similarly, we can find $\delta{\bf\Phi}^{-}\to0,\ \delta{\bf\Phi}^{+}\to{\bf\Phi}^{-}\delta{\bf S}^{-1}$ when $x\to-\infty$, then
\begin{equation}\label{2-delta-S^-1}
	\delta{\bf S}^{-1}(\lambda;t)=-\lambda\int_{-\infty}^{+\infty}\left({\bf\Phi}^{-}(\lambda;x,t)\right)^{-1}\delta{\bf Q}_{x}(x,t){\bf\Phi}^{+}(\lambda;x,t)\dd x.
\end{equation}
As usual, we define the reflection coefficients $\rho_{1}(\lambda;t)$ and $\rho_{2}(\lambda;t)$:
\begin{equation*}
	\rho_{1}(\lambda;t)=\frac{s_{12}(\lambda;t)}{s_{11}(\lambda;t)},\ \ \ \ \rho_{2}(\lambda;t)=\frac{s_{21}(\lambda;t)}{s_{22}(\lambda;t)},\ \ \ \  \lambda\in\Sigma.
\end{equation*}
Combining with the variation of ${\bf S}(\lambda;t)$ (i.e.,\eqref{2-delta-S}) or ${\bf S}^{-1}(\lambda;t)$ (i.e.,\eqref{2-delta-S^-1}), the variations of the reflection coefficients can be derived,
\begin{equation}\label{2-var-qtos}
\begin{split}
		\delta\rho_{1}(\lambda;t)=&\frac{1}{s_{11}^{2}(\lambda;t)}\left(s_{11}(\lambda;t)\delta s_{12}(\lambda;t)-\delta s_{11}(\lambda;t)s_{12}(\lambda;t)\right)
		=\frac{\lambda}{s_{11}^{2}(\lambda;t)}\Big\langle\sigma_{3}{\bf\Omega}_{+,2}(\lambda;x,t),\delta{\bf\widehat u}_{x}(x,t)\Big\rangle,\\
		\delta\rho_{2}(\lambda;t)=&\frac{1}{s_{22}^{2}(\lambda;t)}\left(s_{22}(\lambda;t)\delta s_{21}(\lambda;t)-\delta s_{22}(\lambda;t)s_{21}(\lambda;t)\right)
		=-\frac{\lambda}{s_{22}^{2}(\lambda;t)}\Big\langle\sigma_{3}{\bf\Omega}_{-,1}(\lambda;x,t),\delta{\bf\widehat u}_{x}(x,t)\Big\rangle,
	\end{split}
\end{equation} 
where $\lambda\in\Sigma$, 
\begin{equation*}
{\bf\Omega}_{-,1}=\begin{bmatrix}
		(\phi_{21}^{+})^{2}\\[4pt]
		-(\phi_{11}^{+})^{2}
	\end{bmatrix},\ \ \ \ {\bf\Omega}_{+,2}=\begin{bmatrix}
	(\phi_{22}^{+})^{2}	\\[4pt]
		-(\phi_{12}^{+})^{2}
	\end{bmatrix},\ \ \ \ \delta{\bf \widehat u}_{x}=\begin{bmatrix}
	\delta q_{x}\\[4pt]
	-\delta q_{x}^{*}
\end{bmatrix},
\end{equation*}
and the inner product is defined by
\begin{equation*}
	\langle{\bf f},{\bf g}\rangle=\int_{-\infty}^{+\infty}{\bf f}^{\top}(\lambda;x,t){\bf g}(x,t)\dd x.
\end{equation*}
The functions ${\bf\Omega}_{-,1}(\lambda;x,t)$ and ${\bf\Omega}_{+,2}(\lambda;x,t)$ are called the adjoint squared eigenfunctions of the fFL equation.

Then we discuss the variation of the potential via the variations of the scattering data. Based on the matrix function ${\bf T}(\lambda;x,t)$ which is defined in section \ref{2-sec-IST}, we introduce a new matrix function ${\bf F}(\lambda;x,t)$:
\begin{equation*}
	{\bf F}_{+}={\bf T}_{+}\begin{bmatrix}
		1&0\\[5pt]
		0&\frac{1}{s_{11}}
	\end{bmatrix},\ \  \lambda\in\mathbb{D_{+}},\ \ \ \ \ \ 
{\bf F}_{-}=({\bf T}_{-})^{-1}\begin{bmatrix}
	1&0\\[3pt]
	0&s_{22}
\end{bmatrix},\ \ \lambda\in\mathbb{D_{-}}.
\end{equation*}
Assuming that $s_{11}(\lambda;t)$ and $s_{22}(\lambda;t)$ have no zeros in $\mathbb{D}_{\pm}$, respectively. Then we can construct the following Riemann-Hilbert problem.\\
	{\bf Riemann-Hilbert problem $3$.\ }The matrix function ${\bf F}(\lambda;x,t)$ has the following properties:
\begin{itemize}
	\item ${\bf Analyticity:}$\ ${\bf F}_{\pm}(\lambda;x,t)$ are analytic in $\mathbb{D}_{\pm}$, respectively.  
	\item ${\bf Jump\ condition:}$\ ${\bf F}(\lambda;x,t)$ satisfies the jump condition:
	\begin{equation*}
		{\bf F}_{+}(\lambda;x,t)={\bf F}_{-}(\lambda;x,t)\left(\mathbb{I}-{\bf J}_{F}(\lambda;x,t)\right),\ \ \lambda\in\Sigma,
	\end{equation*}
	where
	\begin{equation*}
		{\bf J}_{F}(\lambda;x,t)=\exp\left(\left(-\ii\lambda^{2}x+\frac{\ii}{4}\mathcal{F}_{fF}(\lambda^{2})t\right){\rm ad}\widehat{\sigma}_{3}\right)
		\begin{bmatrix}
			0&\rho_{1}(\lambda;t)\\[4pt]
			-\rho_{2}(\lambda;t)&\rho_{1}(\lambda;t)\rho_{2}(\lambda;t)
		\end{bmatrix}.
	\end{equation*}
	\item ${\bf Asymptotic\ behavior:}$\
	\begin{equation*}
		{\bf F}(\lambda;x,t)={\bf F}_{0}(x,t)+\mathcal{O}(\lambda^{-1}),\ \ \ \ {\rm as}\ \lambda\to\infty,
	\end{equation*}
where ${\bf F}_{0}(x,t)=\exp\left(\frac{\ii}{2}\int_{-\infty}^{x}|q_{y}(y,t)|^{2}\dd y\sigma_{3}\right)$.
\end{itemize}
\noindent
It can be found that $\det{\bf F}_{\pm}=1$, which means that the Riemann-Hilbert problem $3$ is regular. 
Considering the variation of the jump condition in Riemann-Hilbert problem $3$, we can obtain
\begin{equation}\label{2-delta-RHP}
	\delta{\bf F}_{+}({\bf F}_{+})^{-1}-\delta{\bf F}_{-}({\bf F}_{-})^{-1}=-{\bf F}_{-}\delta{\bf J}_{F}({\bf F}_{+})^{-1},\ \ \ \ \lambda\in\Sigma,
\end{equation}
where
\begin{equation}\label{2-deltaJ}
\begin{split}
	&{\bf F}_{\pm}={\bf F}_{0}+\frac{1}{\lambda}\begin{bmatrix}
		\mp1&-\frac{\ii}{2}q_{x}\\[4pt]
		\frac{\ii}{2}q^{*}_{x}&\pm1
	\end{bmatrix}{\bf F}_{0}+\mathcal{O}(\lambda^{-2}),\ \ \ \ \bar{u}:=\bar{u}(x,t)=-\frac{\ii}{2}\int_{-\infty}^{x}{\bf \widehat u}_{y}^{\top}(y,t)\sigma_{1}\delta{\bf \widehat u}_{y}(y,t)\dd y,\\[4pt]
&\delta{\bf F}_{\pm}=\bar{u}\sigma_{3}{\bf F}_{0}+\frac{\ii}{2\lambda}\begin{bmatrix}
	\pm2\ii&q_{x}-\delta q_{x}\bar{u}^{-1}\\[4pt]
	q^{*}_{x}+\delta q^{*}_{x}\bar{u}^{-1}&\pm2\ii
\end{bmatrix}\bar{u}{\bf F}_{0}+\mathcal{O}(\lambda^{-2}),\\[4pt]
&\delta{\bf J}_{F}=\exp\left(\left(-\ii\lambda^{2}x+\frac{\ii}{4}\mathcal{F}_{fF}(\lambda^{2})t\right){\rm ad}\widehat{\sigma}_{3}\right)
\begin{bmatrix}
	0&\delta\rho_{1}\\
	-\delta\rho_{2}&\rho_{1}\delta\rho_{2}+\rho_{2}\delta\rho_{1}
\end{bmatrix}.
\end{split}
\end{equation}
Obviously, the functions $\delta{\bf F}_{\pm}(\lambda;x,t)\left({\bf F}_{\pm}(\lambda;x,t)\right)^{-1}$ are analytic in $\mathbb{D}_{\pm}$, and equal to $\bar{u}(x,t)\sigma_{3}$ as $\lambda\to\infty$, then the equation \eqref{2-delta-RHP} can be solved by using Plemelj's formulae,
\begin{equation}\label{2-deltaP+P-1}
	\delta{\bf F}_{\pm}({\bf F}_{\pm})^{-1}=\bar{u}(x,t)\sigma_{3}-\frac{1}{2\pi\ii}\int_{\Sigma_{\pm}}\dfrac{{\bf F}_{-}(\zeta;x,t)\delta{\bf J}_{F}(\zeta;x,t)({\bf F}_{+}(\zeta;x,t))^{-1}}{\zeta-(\lambda\pm\ii0)}\dd\zeta.
\end{equation}
On the other hand, the expressions for $\delta{\bf F}_{\pm}(\lambda;x,t)\left({\bf F}_{\pm}(\lambda;x,t)\right)^{-1}$ can also be given in terms of \eqref{2-deltaJ},
\begin{equation}\label{2-deltaP+P-2}
	\delta{\bf F}_{\pm}\left({\bf F}_{\pm}\right)^{-1}=\bar{u}(x,t)\sigma_{3}-\frac{\ii}{2\lambda}\begin{bmatrix}
		0&\delta q_{x}-2q_{x}\bar{u}\\[4pt]
		-\delta q^{*}_{x}-2q_{x}^{*}\bar{u}&0
	\end{bmatrix}+\mathcal{O}(\lambda^{-2}).
\end{equation}
Comparing the equations \eqref{2-deltaP+P-1} and \eqref{2-deltaP+P-2}, there is
\begin{equation*}
	\begin{split}
\begin{bmatrix}
	0&\delta q_{x}-2q_{x}\bar{u}\\[4pt]
	-\delta q^{*}_{x}-2q_{x}^{*}\bar{u}&0
\end{bmatrix}=&\frac{1}{\pi}\int_{\Sigma_{\pm}}{\bf F}_{-}(\zeta;x,t)\delta{\bf J}_{F}(\zeta;x,t)({\bf F}_{+}(\zeta;x,t))^{-1}\dd\zeta\\
=&\frac{1}{\pi}\int_{\Sigma_{\pm}}\left(\delta\rho_{1}(\lambda;t)\begin{bmatrix}
	-\phi_{11}^{-}\phi_{21}^{-}&(\phi_{11}^{-})^{2}\\[4pt]
	-(\phi_{21}^{-})^{2}&\phi_{11}^{-}\phi_{21}^{-}		
\end{bmatrix}+\delta\rho_{2}(\lambda;t)
\begin{bmatrix}
-\phi_{12}^{-}\phi_{22}^{-}&(\phi_{12}^{-})^{2}\\[4pt]
-(\phi_{22}^{-})^{2}&\phi_{12}^{-}\phi_{22}^{-}	
\end{bmatrix}\right)\dd\lambda
	\end{split},
\end{equation*}
in the integral of the above equation, $\phi_{ij}^{-}=\phi_{ij}^{-}(\lambda;x,t),\ i,j=1,2$. Therefore,
\begin{equation}\label{2-var-rhotoq}
	\delta{\bf\widehat u}_{x}(x,t)=\frac{\sigma_{3}}{\pi}\int_{\Sigma_{\pm}}\left(\delta\rho_{1}(\lambda;t){\bf\Omega}_{+,1}(\lambda;x,t)+\delta\rho_{2}(\lambda;t)
{\bf\Omega}_{-,2}(\lambda;x,t)\right)\dd\lambda+2\bar{u}(x,t)\sigma_{3}{\bf\widehat u}_{x}(x,t),
\end{equation}
where ${\bf\Omega}_{+,1}(\lambda;x,t)$ and ${\bf\Omega}_{-,2}(\lambda;x,t)$ are called the squared eigenfunctions of
the fFL equation,
\begin{equation*}
	{\bf\Omega}_{+,1}=\begin{bmatrix}
		(\phi_{11}^{-})^{2}\\[4pt]
		(\phi_{21}^{-})^{2}
	\end{bmatrix},\ \ \ \ {\bf\Omega}_{-,2}=\begin{bmatrix}
		(\phi_{12}^{-})^{2}\\[4pt]
		(\phi_{22}^{-})^{2}
	\end{bmatrix}.
\end{equation*}
Furthermore, we can derive the following relations based on the Lax pair \eqref{2-Lax pair-fFL}: \begin{equation}\label{2-relation-Omega+Omega_x}
 	\begin{split}
	&{\bf\Omega}_{+,1}=\frac{\ii}{2\lambda^{2}}\sigma_{3}\partial_{x}{\bf\Omega}_{+,1}-\frac{\ii}{\lambda}\phi_{11}^{-}\phi_{21}^{-}{\bf\widehat{u}}_{x},\ \ \ \ \phi_{11}^{-}\phi_{21}^{-}=-\lambda\int_{-\infty}^{x}({\bf \widehat u}_{y}^{\top}(y,t)\sigma_{1}\sigma_{3}{\bf\Omega}_{+,1}(\lambda;y,t))\dd y,\\ &{\bf\Omega}_{-,2}=\frac{\ii}{2\lambda^{2}}\sigma_{3}\partial_{x}{\bf\Omega}_{-,2}-\frac{\ii}{\lambda}\phi_{12}^{-}\phi_{22}^{-}{\bf\widehat{u}}_{x},\ \ \ \ \phi_{12}^{-}\phi_{22}^{-}=-\lambda\int_{-\infty}^{x}({\bf \widehat u}_{y}^{\top}(y,t)\sigma_{1}\sigma_{3}{\bf\Omega}_{-,2}(\lambda;y,t))\dd y.
\end{split}
\end{equation}
If we left-multiply both sides of the equation \eqref{2-var-rhotoq} by the vector ${\bf \widehat u}_{x}^{\top}(x,t)\sigma_{1}$ and then integrate the resulting expression, there is
\begin{multline}\label{2-bar-u}
	\int_{-\infty}^{x}{\bf \widehat u}_{y}^{\top}(y,t)\sigma_{1}\delta{\bf\widehat u}_{y}(y,t)\dd y\\
	=\frac{1}{\pi}\int_{-\infty}^{x}\int_{\Sigma_{\pm}}\left(\delta\rho_{1}(\lambda;t){\bf \widehat u}_{y}^{\top}(y,t)\sigma_{1}\sigma_{3}{\bf\Omega}_{+,1}(\lambda;y,t)+\delta\rho_{2}(\lambda;t){\bf \widehat u}_{y}^{\top}(y,t)\sigma_{1}\sigma_{3}
	{\bf\Omega}_{-,2}(\lambda;y,t)\right)\dd\lambda\dd y,
\end{multline}
which allows us to rewrite the function $\bar{u}(x,t)$. Substituting \eqref{2-relation-Omega+Omega_x}, \eqref{2-bar-u} and \eqref{2-var-qtos} into \eqref{2-var-rhotoq}, then
\begin{equation}\label{2-deltafunction-Y}
	\delta{\bf\widehat u}_{x}(x,t)=\int_{-\infty}^{+\infty}{\bf Y}(x,y,t)\delta{\bf \widehat u}_{y}(y,t)\dd y,
\end{equation}
where
\begin{equation*}
	{\bf Y}(x,y,t){=}\frac{\ii}{2\pi}\int_{\Sigma_{\pm}}\frac{1}{\lambda}\left(\frac{1}{ s_{11}^{2}(\lambda;t)}\left(\partial_{x}{\bf\Omega}_{+,1}(\lambda;x,t)\right){\bf\Omega}_{+,2}^{\top}(\lambda;y,t){-}
	\frac{1}{s_{22}^{2}(\lambda;t)}\left(\partial_{x}{\bf\Omega}_{-,2}(\lambda;x,t)\right){\bf\Omega}_{-,1}^{\top}(\lambda;y,t)\right)\sigma_{3}\dd\lambda.
\end{equation*}
If we want to ensure the equation \eqref{2-deltafunction-Y} is always true, regardless of the localized function $\delta{\bf\widehat u}_{x}(x,t)$ we choose, then the function ${\bf Y}(x,y,t)$ must be the Dirac's $\delta$-function, that is, 
\begin{equation}\label{2-closure-nozero}
	{\bf Y}(x,y,t)=\delta(x-y)\mathbb{I},
\end{equation}
which is called the closure relation or the completeness relation of the squared eigenfunctions. Note that the completeness relation \eqref{2-closure-nozero} was obtained based on the assumption that $s_{11}(\lambda;t)$ and $s_{22}(\lambda;t)$ have no zeros. However, the functions $s_{11}(\lambda;t)$ and $s_{22}(\lambda;t)$ have simple zeros in $\Lambda_{1}$ and $\Lambda_{2}$, respectively. Therefore, the completeness relation \eqref{2-closure-nozero} no longer holds and we must take into account contributions from the discrete spectrums. Furthermore, it has been ascertained that these contributions from the discrete spectrums are precisely represented by the residues of the integrand functions in \eqref{2-closure-nozero} at the poles $\lambda_{n}$ (i.e., the zeros of $s_{11}(\lambda;t)$) and $\lambda_{n}^{*}$ (i.e., the zeros of $s_{22}(\lambda;t)$), $n=1,2,\cdots,2N$. By incorporating these contributions from the discrete spectrums, we can establish a comprehensive completeness relation:
\begin{equation*}
	{\bf\widehat Y}(x,y,t)=\delta(x-y)\mathbb{I},
\end{equation*}
where
\begin{multline*}
{\bf\widehat Y}(x,y,t)=-\frac{\ii}{2\pi}\bigg(\int_{\Gamma_{+}}\frac{1}{\lambda s_{11}^{2}(\lambda;t)}\left(\partial_{x}{\bf\Omega}_{+,1}(\lambda;x,t)\right){\bf\Omega}_{+,2}^{\top}(\lambda;y,t)\sigma_{3}\dd\lambda\\
	-\int_{\Gamma_{-}}\frac{1}{\lambda s_{22}^{2}(\lambda;t)}\left(\partial_{x}{\bf\Omega}_{-,2}(\lambda;x,t)\right){\bf\Omega}_{-,1}^{\top}(\lambda;y,t)\sigma_{3}\dd\lambda\bigg),
\end{multline*}
and the integration paths $\Gamma_{\pm}$ are shown in Fig.\ref{fig:contour}. Based on the property of the Dirac's $\delta$-function,
\begin{equation*}
	\int_{-\infty}^{+\infty}\delta(x-y)f(y)\dd y=f(x),
\end{equation*}
where $f(x)$ is an arbitrary smooth function, we have
\begin{multline}\label{2-hatu-usesquared}
	{\bf\widehat u}_{x}(x,t)=\int_{-\infty}^{+\infty}{\bf\widehat Y}(x,y,t){\bf \widehat u}_{y}(y,t)\dd y
	=-\frac{\ii}{2\pi}\bigg(\int_{\Gamma_{+}}\frac{1}{\lambda s_{11}^{2}(\lambda;t)}\left(\partial_{x}{\bf\Omega}_{+,1}(\lambda;x,t)\right)\Big\langle\sigma_{3}{\bf\Omega}_{+,2}(\lambda;x,t),{\bf\widehat u}_{x}(x,t)\Big\rangle\dd\lambda\\
	 -\int_{\Gamma_{-}}\frac{1}{\lambda s_{22}^{2}(\lambda;t)}\left(\partial_{x}{\bf\Omega}_{-,2}(\lambda;x,t)\right)\Big\langle\sigma_{3}{\bf\Omega}_{-,1}(\lambda;x,t),{\bf\widehat u}_{x}(x,t)\Big\rangle\dd\lambda\bigg).
\end{multline}

Through a direct calculation, we can derive
\begin{equation*}
\begin{split}
		&\mathcal{L}(\sigma_{3}\partial_{x}{\bf\Omega}_{+,1})=-\lambda^{2}\sigma_{3}\partial_{x}{\bf\Omega}_{+,1},\ \ \ \ \mathcal{L}(\sigma_{3}\partial_{x}{\bf\Omega}_{-,2})=-\lambda^{2}\sigma_{3}\partial_{x}{\bf\Omega}_{-,2},\\
		&\widetilde{\mathcal{L}}(\sigma_{3}\partial_{x}{\bf\Omega}_{-,1})=-\lambda^{2}\sigma_{3}\partial_{x}{\bf\Omega}_{-,1},\ \ \ \ \widetilde{\mathcal{L}}(\sigma_{3}\partial_{x}{\bf\Omega}_{+,2})=-\lambda^{2}\sigma_{3}\partial_{x}{\bf\Omega}_{+,2},
	\end{split}
\end{equation*}
here the recursion operator $\mathcal{L}$ corresponds to the fFL equation, and $\widetilde{\mathcal{L}}$ is its adjoint. Then we can generalize it to
\begin{equation}\label{2-FactOmega}
\begin{split}
			&\mathcal{F}_{fF}(\mathcal{L})(\sigma_{3}\partial_{x}{\bf\Omega}_{+,1})=-\frac{|2\lambda^{2}|^{2\epsilon}}{\lambda^{2}}\sigma_{3}\partial_{x}{\bf\Omega}_{+,1},\ \ \ \ 	\mathcal{F}_{fF}(\mathcal{L})(\sigma_{3}\partial_{x}{\bf\Omega}_{-,2})=-\frac{|2\lambda^{2}|^{2\epsilon}}{\lambda^{2}}\sigma_{3}\partial_{x}{\bf\Omega}_{-,2},\\
			&\mathcal{F}_{fF}(\widetilde{\mathcal{L}})(\sigma_{3}\partial_{x}{\bf\Omega}_{-,1})=-\frac{|2\lambda^{2}|^{2\epsilon}}{\lambda^{2}}\sigma_{3}\partial_{x}{\bf\Omega}_{-,1},\ \ \ \ \mathcal{F}_{fF}(\widetilde{\mathcal{L}})(\sigma_{3}\partial_{x}{\bf\Omega}_{+,2})=-\frac{|2\lambda^{2}|^{2\epsilon}}{\lambda^{2}}\sigma_{3}\partial_{x}{\bf\Omega}_{+,2}.
	\end{split}
\end{equation}
Utilizing the relation \eqref{2-FactOmega}, we apply the operator function $\mathcal{F}_{fF}(\mathcal{L})$ to the vector function ${\bf\widehat u}_{x}(x,t)$ (i.e.,\eqref{2-hatu-usesquared}),
\begin{multline*}
	\mathcal{F}_{fF}(\mathcal{L}){\bf\widehat u}_{x}(x,t)=-\dfrac{1}{2\pi\ii}\bigg(\int_{\Gamma_{+}}\dfrac{|2\lambda^{2}|^{2\epsilon}}{\lambda^{3} s_{11}^{2}(\lambda;t)}\left(\partial_{x}{\bf\Omega}_{+,1}(\lambda;x,t)\right)\Big\langle\sigma_{3}{\bf\Omega}_{+,2}(\lambda;y,t),{\bf\widehat u}_{y}(y,t)\Big\rangle\dd\lambda\\
	-\int_{\Gamma_{-}}\dfrac{|2\lambda^{2}|^{2\epsilon}}{\lambda^{3} s_{22}^{2}(\lambda;t)}\  \left(\partial_{x}{\bf\Omega}_{-,2}(\lambda
	;x,t)\right)\Big\langle\sigma_{3}{\bf\Omega}_{-,1}(\lambda;y,t),{\bf\widehat u}_{y}(y,t)\Big\rangle\dd\lambda\bigg).
\end{multline*}
Hence, based on the equation \eqref{2-evolution eq}, we can deduce the exact form of the fFL equation,
\begin{equation}\label{2-exact-fFL}
	\begin{split}
		q_{xt}(x,t)=&-\frac{1}{4\pi}\int_{\Gamma_{+}}\dfrac{|2\lambda^{2}|^{2\epsilon}}{\lambda^{3} s_{11}^{2}(\lambda;t)}\left(\partial_{x}\left((\phi_{11}^{-}(\lambda;x,t))^2\right)\right)
		\int_{-\infty}^{+\infty}\Big(\big(\phi_{12}^{+}(\lambda;y,t)\big)^{2}q^{*}_{y}(y,t){-}\big(\phi_{22}^{+}(\lambda;y,t)\big)^{2}q_{y}(y,t)\Big)\dd y\dd\lambda\\
		&{+}\frac{1}{4\pi}\int_{\Gamma_{-}}\dfrac{|2\lambda^{2}|^{2\epsilon}}{\lambda^{3} s_{22}^{2}(\lambda;t)}\left(\partial_{x}\left((\phi_{12}^{-}(\lambda;x,t))^2\right)\right)
		\int_{-\infty}^{+\infty}\Big(\big(\phi_{11}^{+}(\lambda;y,t)\big)^{2}q^{*}_{y}(y,t){-}\big(\phi_{21}^{+}(\lambda;y,t)\big)^{2}q_{y}(y,t)\Big)\dd y\dd\lambda.
	\end{split}
\end{equation}
In particular, the equation \eqref{2-exact-fFL} will degenerate into the classical FL equation \eqref{2-FL} when $\epsilon=0$.

\section{Fractional $N$-soliton solution}
In this section, we will consider the fractional $N$-soliton solution of the fFL equation. This can be attained by setting $s_{12}(\lambda;t)=s_{21}(\lambda;t)=0$. Consequently, ${\bf J}(\lambda;x,t)=\mathbb{I}$,
\begin{equation*}
	{\bf M}_{+}(\lambda;x,t)={\bf M}_{+}^{[1]}(\lambda;x,t){\bf G}(\lambda;x,t)={\bf E}(x,t){\bf G}(\lambda;x,t).
\end{equation*}
We assume ${\bf G}(\lambda;x,t)={\bf G}^{(0)}(x,t)+\lambda{\bf G}^{(1)}(x,t)+o(\lambda)$ by combining with ${\bf M}_{+,0}(x,t)=\mathbb{I}$, then
\begin{equation}\label{2-q-recover}
	{\bf G}^{(0)}(x,t)=\left({\bf E}(x,t)\right)^{-1},\ \ \ \ q(x,t)=\left({\bf E}(x,t){\bf G}^{(1)}(x,t)\right)_{12}.
\end{equation}

Next, we need to derive the exact form of ${\bf E}(x,t){\bf G}^{(1)}(x,t)$. For the matrix function ${\bf G}(\lambda;x,t)$, it can be rewritten as follows:
\begin{equation*}
	{\bf G}=\mathbb{I}+\sum_{j=1}^{N}\left(\frac{{\bf B}_{j}}{\lambda-\lambda_{j}^{*}}-\frac{\sigma_{3}{\bf B}_{j}\sigma_{3}}{\lambda+\lambda_{j}^{*}}\right).
\end{equation*}
Note that ${\bf E}(x,t){\bf G}(\lambda;x,t)$ has the poles at $\lambda=\pm\lambda_{j}^{*},\ j=1,\cdots,N$, and tends to identity matrix $\mathbb{I}$ as $\lambda\to 0$. By introducing the variables $\kappa=\frac{1}{\lambda},\ \kappa_{j}=\frac{1}{\lambda_{j}}$, we can assume 
\begin{equation*}
		{\bf EG}=\mathbb{I}+\sum_{j=1}^{N}\left(\frac{{\bf C}_{j}}{\kappa-\kappa_{j}^{*}}-\frac{\sigma_{3}{\bf C}_{j}\sigma_{3}}{\kappa+\kappa_{j}^{*}}\right),
\end{equation*}
which satisfies ${\bf E}(x,t){\bf G}(\kappa;x,t)\to\mathbb{I}$ as $\kappa\to+\infty$. Therefore,
\begin{equation*}
	{\bf EG}^{(1)}=\sum_{j=1}^{N}\left({\bf C}_{j}-\sigma_{3}{\bf C}_{j}\sigma_{3}\right).
\end{equation*}

We define the matrix ${\bf C}_{j}(x,t)$ as $|{\bf y}_{j}\rangle\langle{\bf z}_{j}|$, where $|{\bf y}_{j}\rangle=\begin{bmatrix}
 	y_{1j},&y_{2j}
 \end{bmatrix}^{\top},\ |{\bf z}_{j}\rangle=\begin{bmatrix}
 z_{1j},&z_{2j}
\end{bmatrix}^{\top}$. Based on the equation \eqref{2-q-recover}, then there is
\begin{equation*}
	q(x,t)=2\langle{\bf\widehat{z}}_{2}|{\bf{\widehat y}}_{1}\rangle,
\end{equation*}
where
\begin{equation*}
	|\widehat{{\bf y}}_{k}\rangle=\begin{bmatrix}
		y_{k1},&\cdots,&y_{kN}
	\end{bmatrix}^{\top},\ \ \ \ |\widehat{{\bf z}}_{k}\rangle=\begin{bmatrix}
	z_{k1},&\cdots,&z_{kN}
\end{bmatrix}^{\top},\ \ \ \ k=1,2.
\end{equation*}
In addition, we can easily give the inverse of ${\bf E}(x,t){\bf G}(\kappa;x,t)$,
\begin{equation*}
	({\bf EG})^{-1}=\mathbb{I}+\sum_{j=1}^{N}\left(\frac{{\bf C}_{j}^{\dagger}}{\kappa-\kappa_{j}}-\frac{\sigma_{3}{\bf C}_{j}^{\dagger}\sigma_{3}}{\kappa+\kappa_{j}}\right).
\end{equation*}
It is obvious that $\underset{\kappa=\kappa_{l}}{\rm{Res}}\left({\bf E}(x,t){\bf G}(\kappa;x,t)\left({\bf E}(x,t){\bf G}(\kappa;x,t)\right)^{-1}\right)=0$, then
\begin{equation}\label{2-res-EG}
	\left(\mathbb{I}+\sum_{j=1}^{N}\left(\frac{|{\bf y}_{j}\rangle\langle{\bf z}_{j}|}{\kappa_{l}-\kappa_{j}^{*}}-\frac{\sigma_{3}|{\bf y}_{j}\rangle\langle{\bf z}_{j}|\sigma_{3}}{\kappa_{l}+\kappa_{j}^{*}}\right)\right)|{\bf z}_{l}\rangle=0,\ \ \ \ l=1,\cdots,N.
\end{equation}
The above equation can be rewritten as
\begin{equation*}
	|{\bf z}_{l}\rangle+\sum_{j=1}^{N}\left(\frac{\langle{\bf z}_{j}|{\bf z}_{l}\rangle}{\kappa_{l}-\kappa_{j}^{*}}\mathbb{I}-\frac{\langle{\bf z}_{j}|\sigma_{3}|{\bf z}_{l}\rangle}{\kappa_{l}+\kappa_{j}^{*}}\sigma_{3}\right)|{\bf y}_{j}\rangle=0,\ \ \ \ l=1,\cdots,N,
\end{equation*}
which implies $|\widehat{{\bf y}}_{1}\rangle={\bf N}^{-1}|\widehat{{\bf z}}_{1}\rangle$, ${\bf N}(\kappa;x,t)=\left(n_{jk}\right)_{1\leq j,k\leq N}$,
\begin{equation*}
	n_{jk}=\left(\frac{\langle{\bf z}_{k}|\sigma_{3}|{\bf z}_{j}\rangle}{\kappa_{j}+\kappa_{k}^{*}}-\frac{\langle{\bf z}_{k}|{\bf z}_{j}\rangle}{\kappa_{j}-\kappa_{k}^{*}}\right)_{1\leq j,k\leq N}.
\end{equation*}
Then we can give the potential function in determinant form,
\begin{equation*}
	q(x,t)=2\langle{\bf\widehat{z}}_{2}|{\bf{\widehat y}}_{1}\rangle=2\langle{\bf\widehat{z}}_{2}|{\bf N}^{-1}|\widehat{{\bf z}}_{1}\rangle=-2\frac{\det\widehat{\bf N}}{\det{\bf N}},\ \ \ \ \widehat{\bf N}=\begin{bmatrix}
		{\bf N}&|{\bf{\widehat z}}_{1}\rangle\\[2pt]
		\langle{\bf\widehat{z}}_{2}|&0
	\end{bmatrix}.
\end{equation*}

Based on the equations \eqref{2-T-kernel} and \eqref{2-res-EG}, we can choose $|{\bf z}_{l}\rangle=|{\bf v}_{l}\rangle$. Next we will give the explicit form of the vector $|{\bf v}_{j}\rangle$. By taking the $x$-derivative of
both sides of the first equation in \eqref{2-T-kernel} and solving this equation, we can obtain
\begin{equation*}
	|{\bf v}_{j}(\lambda_{j};x)\rangle=\ee^{-\ii\lambda_{j}^{2}\sigma_{3}x+\int^{x}_{x_{0}}\zeta_{j}(y)dy}|\nu_{j}\rangle,\ \ \ \ |\nu_{j}\rangle=\begin{bmatrix}
		\nu_{1j}&\nu_{2j}
	\end{bmatrix}^{\top},
\end{equation*}
where $|\nu_{j}\rangle$ is independent of $x$ and $t$. Without losing generality, we take $\zeta_{j}(y)=0$, then $|{\bf v}_{j}(\lambda_{j};x)\rangle=\ee^{-\ii\lambda_{j}^{2}\sigma_{3}x}|\nu_{j}\rangle$. Combining the time evolution of $|{\bf v}_{j}(\lambda_{j};x)\rangle$, then
\begin{equation*}
	|{\bf v}_{j}(\lambda_{j};x,t)\rangle=\ee^{\left(-\ii\lambda_{j}^{2}x+\frac{\ii}{4}\mathcal{F}_{fF}(\lambda_{j}^{2})t\right)\sigma_{3}}|\nu_{j}\rangle.
\end{equation*}
For convenience, we denote
\begin{equation}\label{2-v_j}
	|{\bf v}_{j}\rangle=\begin{bmatrix}
		\ee^{\theta_{j}}\nu_{1j}&\ee^{-\theta_{j}}\nu_{2j}
	\end{bmatrix}^{\top},\ \ \ \ \theta_{j}=\theta_{j}(\lambda_{j};x,t,\epsilon)=-\ii\lambda_{j}^{2}x+\frac{\ii}{4}\lambda_{j}^{-2}|2\lambda_{j}^{2}|^{2\epsilon}t.
\end{equation}

Combining with the form of $|{\bf v}_{j}\rangle$ (i.e.,\eqref{2-v_j}), and the transformation between $\kappa$ and $\lambda$, then the fractional $N$-soliton solution of the fFL equation can be represented as
\begin{equation}\label{2-recover-qx}
	q^{[N]}(x,t)=-2\frac{\det\widehat{\bf P}}{\det{\bf P}},\ \ \ \ {\widehat{\bf P}}=\begin{bmatrix}
		{\bf P}&\hat{{\bf p}}_{1}\\[3pt]
		\hat{{\bf p}}_{2}&0
	\end{bmatrix},
\end{equation}
where ${\bf P}(\lambda;x,t)=\left(p_{jk}\right)_{1\leq j,k\leq N}$,
\begin{equation*}
\hat{{\bf p}}_{1}=\begin{bmatrix}
		\ee^{\theta_{1}}\nu_{11},&\cdots,&\ee^{\theta_{N}}\nu_{1N}
	\end{bmatrix}^{\top},\ \ \ \ 
\hat{{\bf p}}_{2}=\begin{bmatrix}
\ee^{-\theta_{1}^{*}}\nu_{21}^{*},&\cdots,&\ee^{-\theta_{N}^{*}}\nu_{2N}^{*}
\end{bmatrix},\ \ \ \ 
	p_{jk}=\frac{2\lambda_{j}\lambda_{k}^{*}}{\lambda_{j}^{2}-\lambda_{k}^{*2}}\langle {\bf v}_{k}|\begin{bmatrix}
		\lambda_{j}&0\\
		0&\lambda_{k}^{*}
	\end{bmatrix}|{\bf v}_{j}\rangle,
\end{equation*}
the notation $\langle {\bf v}_{k}|$ represents the conjugate transpose of the vector $|{\bf v}_{k}\rangle$.

\subsection{Fractional one-soliton solution}

When considering the case of $N=1$, we can choose $\lambda_{1}=\xi+\ii\eta$, resulting in $\lambda_{2}=-\xi-\ii\eta$. Utilizing equation \eqref{2-recover-qx}, the fractional one-soliton solution can be obtained and presented in the subsequent proposition. It will be elucidated that this solution effectively provides a resolution to the fFL equation \eqref{2-exact-fFL}.
\begin{prop}
	The fractional one-soliton solution of the fFL equation \eqref{2-exact-fFL} is as follows:
	\begin{equation}\label{2-onesoliton}
		q^{[1]}(x,t)=\frac{2\ii\xi\eta\nu_{21}^{*}\ee^{2\ii\theta_{1I}+\varpi_{1}}}{(\xi+\ii\eta)^{2}(\xi-\ii\eta)\nu_{11}^{*}}{\rm sech}(2\theta_{1R}+\varpi_{1}),
	\end{equation}
where $\nu_{11},\ \nu_{21}$ are complex constants,  $\theta_{1}=\theta_{1R}+\ii\theta_{1I}$,
\begin{equation*}
	\begin{split}
		&\theta_{1R}=2\xi\eta x+2^{2\epsilon-1}\xi\eta(\xi^2+\eta^2)^{2\epsilon-2} t,\ \ \ \ \varpi_{1}=\ln\left(\frac{(\xi+\ii\eta)\ |\nu_{11}|}{\sqrt{\xi^2+\eta^2}|\nu_{21}|}\right),\\
		&\theta_{1I}=-(\xi^2-\eta^2)x+4^{\epsilon-1}(\xi^2-\eta^2)(\xi^2+\eta^2)^{2\epsilon-2}t.
	\end{split}	
\end{equation*}
\end{prop}
\begin{proof}
	
	In order to establish the validity of the fractional one-soliton solution $q^{[1]}(x,t)$ in satisfying equation \eqref{2-exact-fFL}, it is imperative to obtain the corresponding fundamental solution ${\bf\Phi}^{[1]}(\lambda;x,t)$ as well as the scattering coefficients 
$s_{11}(\lambda;t)$ and $s_{22}(\lambda;t)$. Upon thorough deliberation, we propose utilizing the Darboux transformation method to construct ${\bf\Phi}^{[1]}(\lambda;x,t)$, which offers a more convenient approach. Based on the literature \cite{guo-2012}, we make appropriate modifications to its one-fold Darboux matrix, resulting in
	\begin{equation*}
		{\bf D}=\mathbb{I}-\lambda\lambda_{1}^{*}\left(\dfrac{{\bf B}}{\lambda-\lambda_{1}^{*}}+\dfrac{\sigma_{3}{\bf B}\sigma_{3}}{\lambda+\lambda_{1}^{*}}\right),
	\end{equation*}
	where
	\begin{equation*}
		\begin{split}
			&{\bf B}=\dfrac{\lambda_{1}^{2}-\lambda_{1}^{*2}}{2|\lambda_{1}|^{2}}\begin{bmatrix}
				b&0\\[3pt]
				0&b^{*}
			\end{bmatrix}\varphi_{1}\varphi_{1}^{\dagger},\ \ \ \ b^{-1}=\varphi_{1}^{\dagger}\begin{bmatrix}
				\lambda_{1}&0\\[3pt]
				0&\lambda_{1}^{*}
			\end{bmatrix}\varphi_{1},\ \ \ \ \varphi_{1}=\begin{bmatrix}
				\nu_{11}\exp\left(-\ii\lambda_{1}^{2}x+\frac{\ii}{4}\lambda_{1}^{-2}|2\lambda_{1}^{2}|^{2\epsilon}t\right)\\[5pt]
				\nu_{21}\exp\left(\ii\lambda_{1}^{2}x-\frac{\ii}{4}\lambda_{1}^{-2}|2\lambda_{1}^{2}|^{2\epsilon}t\right)
			\end{bmatrix}.
		\end{split}
	\end{equation*}
For convenience, we denote $\widehat{\bf\Phi}:={\bf\Phi}^{[1]}$, then $\widehat{\bf\Phi}={\bf D}{\bf\Phi}^{[0]}$ is shown below,
		\begin{equation*}
		\begin{split}
			&\widehat{\bf\Phi}=\dfrac{1}{(\xi+\ii\eta)(\lambda+\xi-\ii\eta)(\lambda-\xi+\ii\eta)}\times\\[8pt]
			&
			\begin{bmatrix}
				(\xi{-}\ii\eta)\left(\lambda^{2}\cosh(2\theta_{1R}{+}\varpi_{1}^{*}){\rm sech}(2\theta_{1R}{+}\varpi_{1}){-}(\xi^{2}{+}\eta^{2})\right)\ee^{\theta}&-2\ii\lambda\xi\eta\frac{\nu_{21}^{*}}{\nu_{11}^{*}}\exp(2\ii\theta_{1I}+\varpi_{1}^{*}){\rm sech}(2\theta_{1R}+\varpi_{1})\ee^{{-}\theta}\\[10pt]
				-2\ii\lambda\xi\eta\frac{\nu_{21}}{\nu_{11}}\exp(-2\ii\theta_{1I}+\varpi_{1}^{*}){\rm sech}(2\theta_{1R}+\varpi_{1}^{*})\ee^{\theta}&(\xi{-}\ii\eta)\left(\lambda^{2}\cosh(2\theta_{1R}{+}\varpi_{1}){\rm sech}(2\theta_{1R}{+}\varpi_{1}^{*}){-}(\xi^{2}{+}\eta^{2})\right)\ee^{{-}\theta}
			\end{bmatrix},\\[3pt]
		&{\bf\Phi}^{[0]}={\rm diag}\left(\ee^{\theta},\ee^{-\theta}\right),\ \ \ \ \theta(\lambda;x,t,\epsilon)=-\ii\lambda^{2}x+\frac{\ii}{4}\lambda^{-2}|2\lambda^{2}|^{2\epsilon}t,
		\end{split}
	\end{equation*}
 where ${\bf\Phi}^{[0]}(\lambda;x,t)$ is the fundamental solution under the zero background. We can also derive $\widehat{\bf\Phi}^{\pm}(\lambda;x,t)$ by combining the boundary condition \eqref{2-asy-Phi},
		\begin{equation}\label{2-hatPhi}
			\widehat{\bf\Phi}^{+}=\widehat{\bf\Phi}\begin{bmatrix}
				\dfrac{(\xi+\ii\eta)^{2}(\lambda+\xi-\ii\eta)(\lambda-\xi+\ii\eta)}{(\xi-\ii\eta)^{2}(\lambda-\xi-\ii\eta)(\lambda+\xi+\ii\eta)}&0\\[10pt]
				0&1
			\end{bmatrix},\ \ \ \  \widehat{\bf\Phi}^{-}=\widehat{\bf\Phi}\begin{bmatrix}
				1&0\\[10pt]
				0&\dfrac{(\xi+\ii\eta)^{2}(\lambda+\xi-\ii\eta)(\lambda-\xi+\ii\eta)}{(\xi-\ii\eta)^{2}(\lambda-\xi-\ii\eta)(\lambda+\xi+\ii\eta)}
			\end{bmatrix}.
	\end{equation}
Furthermore, the corresponding scattering coefficients are as follows:
\begin{equation}\label{2-pf-s11+s22}
	s_{11}(\lambda;t)=\dfrac{(\xi-\ii\eta)^{2}(\lambda-\xi-\ii\eta)(\lambda+\xi+\ii\eta)}{(\xi+\ii\eta)^{2}(\lambda+\xi-\ii\eta)(\lambda-\xi+\ii\eta)},\ \ \ \ s_{22}(\lambda;t)=\dfrac{(\xi+\ii\eta)^{2}(\lambda+\xi-\ii\eta)(\lambda-\xi+\ii\eta)}{(\xi-\ii\eta)^{2}(\lambda-\xi-\ii\eta)(\lambda+\xi+\ii\eta)}.
\end{equation}
	Based on the above, we begin to prove that $q^{[1]}(x,t)$ satisfies the equation \eqref{2-exact-fFL}.
	
	Firstly, we consider the integral
	\begin{equation}\label{2-integral-gamma+z}
	\int_{-\infty}^{+\infty}\left(\left(\widehat{\bf\phi}_{12}^{+}(\lambda;x,t)\right)^{2}q^{[1]*}_{x}(x,t){-}\left(\widehat{\bf\phi}_{22}^{+}(\lambda;x,t)\right)^{2}q^{[1]}_{x}(x,t)\right)\dd x:=\int_{-\infty}^{+\infty}g_{+,1}(\lambda;x,t)\dd x.
	\end{equation}
The explicit form of $g_{+,1}(\lambda;x,t)$ can be obtained	by using \eqref{2-onesoliton} and \eqref{2-hatPhi},
\begin{multline*}
	g_{+,1}(\lambda;x,t)=-\dfrac{4\xi\eta(\xi-\ii\eta)^2\nu_{21}^{*}\nu_{11}^{*-1}\ee^{-2\theta+2\ii\theta_{1I}+\varpi_{1}}}{(\xi+\ii\eta)^{3}\left(\lambda^{2}-(\xi-\ii\eta)^{2}\right)^{2}}\bigg(\lambda^{4}{\rm sech}(2\theta_{1R}+\varpi_{1}^{*})-2\lambda^{2}(\xi^2+\eta^2){\rm sech}(2\theta_{1R}+\varpi_{1})\\
+\frac{4\xi^2\eta^2\lambda^2}{\xi^2+\eta^2}{\rm sech}(2\theta_{1R}+\varpi_{1}){\rm sech}^{2}(2\theta_{1R}+\varpi_{1}^{*})+(\xi^2+\eta^2)^{2}{\rm cosh}(2\theta_{1R}+\varpi_{1}^{*}){\rm sech}^{2}(2\theta_{1R}+\varpi_{1})\bigg).
\end{multline*}
By introducing the variable $z=2\theta_{1R}+\varpi_{1R}$, then the integral \eqref{2-integral-gamma+z} can be rewritten as
\begin{equation*}
	\int_{-\infty}^{+\infty}g_{+,1}(\lambda;x,t)\dd x=-\frac{(\xi-\ii\eta)^2\nu_{21}^{*}\nu_{11}^{*-1}\ee^{\varpi_{1}-c(\lambda;t)}}{(\xi+\ii\eta)^3\left(\lambda^{2}-(\xi-\ii\eta)^{2}\right)^{2}}	\int_{-\infty}^{+\infty}h_{1}(\lambda;z,t)\dd z,
\end{equation*}
where $\varpi_{1}=\varpi_{1R}+\ii\varpi_{1I}$,
\begin{equation*}
	\begin{split}
	&c(\lambda;t)=\ii\left( 2^{2\epsilon-1}|\lambda^{2}|^{2\epsilon}\lambda^{-2}+4(\lambda^{2}-2\xi^{2}+2\eta^{2})(2(\xi^2+\eta^2))^{2\epsilon-2}\right)t+\frac{\ii(\lambda^{2}-\xi^2+\eta^2)\varpi_{1R}}{2\xi\eta},\\
&h_{1}(\lambda;z,t)=\exp\left(\frac{\ii(\lambda^{2}-\xi^{2}+\eta^{2})}{2\xi\eta}z\right)\Big(\lambda^{4}{\rm sech}(z-\ii\varpi_{1I})+(\xi^2+\eta^2)^{2}{\rm cosh}(z-\ii\varpi_{1I}){\rm sech}^{2}(z+\ii\varpi_{1I})\\
&\ \ \ \ \ \ \ \ \ \ \ \ \ \ \  -2\lambda^{2}(\xi^2+\eta^2){\rm sech}(z+\ii\varpi_{1I})+\frac{4\xi^2\eta^2\lambda^2}{\xi^2+\eta^2}{\rm sech}(z+\ii\varpi_{1I}){\rm sech}^{2}(z-\ii\varpi_{1I})\Big),\\
&\varpi_{1R}=\ln|\nu_{11}\nu_{21}^{-1}|,\ \ \ \ \varpi_{1I}={\rm arctan}(\eta\xi^{-1}).
\end{split}
\end{equation*}
Now we consider the integral of $h_{1}(\lambda;z,t)$ on the matrix contour in the complex $z$-plane (see Fig.\ref{fig:z-plane}).
	\begin{figure}
		\centering
		\includegraphics[width=0.45\linewidth]{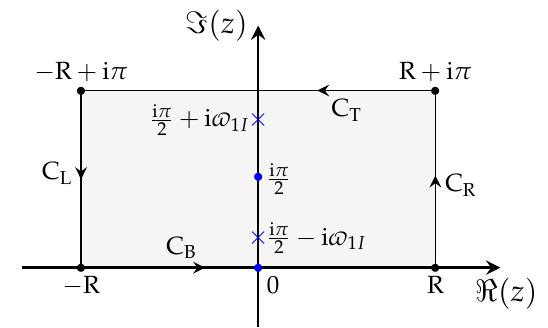}
		\caption[z-plane]{The matrix contour in the complex $z$-plane, which contains only two singularities of the function $h_{1}(\lambda;z,t)$. The singularities of $h_{1}(\lambda;z,t)$ are indicated by $\times$’s in the figure.}
		\label{fig:z-plane}
	\end{figure}
 Apparently, within the closed path shown in Fig.\ref{fig:z-plane}, $h_{1}(\lambda;z,t)$ has singularities at $z=\pm\ii\varpi_{1I}+\frac{\ii\pi}{2}$, so
 \begin{equation*}
 	\int_{C_{B}}h_{1}(\lambda;z,t)\dd z+\int_{C_{R}}h_{1}(\lambda;z,t)\dd z+\int_{C_{T}}h_{1}(\lambda;z,t)\dd z+\int_{C_{L}}h_{1}(\lambda;z,t)\dd z=2\pi\ii\underset{z=\pm\ii\varpi_{1I}+\frac{\ii\pi}{2}}{\rm{Res}}h_{1}(\lambda;z,t).
 \end{equation*}
Through calculations, we find that when $\lambda\in\Sigma_{+}$ or  $\lambda=\pm\lambda_{1}$, there is
\begin{equation*}
\lim\limits_{R\to+\infty}\left(\int_{C_{R}}h_{1}(\lambda;z,t)\dd z+\int_{C_{L}}h_{1}(\lambda;z,t)\dd z\right)=0.
\end{equation*}
In addition, the integral on path $C_{T}$ can be expressed by the integral on path $C_{B}$,
	\begin{equation*}
\int_{C_{T}}h_{1}(\lambda;z,t)\dd z=\exp\left(\frac{\pi(\lambda^{2}-\xi^{2}+\eta^{2})}{2\xi\eta}\right)\int_{C_{B}}h_{1}(\lambda;z,t)\dd z.
	\end{equation*}
Therefore, we can get
	\begin{equation*}
	\left(1+\exp\left(\frac{\pi(\lambda^{2}-\xi^{2}+\eta^{2})}{2\xi\eta}\right)\right)	\int_{-\infty}^{+\infty}h_{1}(\lambda;z,t)\dd z=2\pi\ii\underset{z=\pm\ii\varpi_{1I}+\frac{\ii\pi}{2}}{\rm{Res}}h_{1}(\lambda;z,t),\ \ {\rm when}\ \lambda\in\Sigma_{+}\ {\rm or}\ \lambda=\pm\lambda_{1}.
	\end{equation*}
However, the residues of $h_{1}(\lambda;z,t)$ at $z=\pm\ii\varpi_{1I}+\frac{\ii\pi}{2}$ are all equal to zero, thus $\int_{-\infty}^{+\infty}h_{1}(\lambda;z,t)\dd z=0$ when $\lambda\in\Sigma_{+}$ or  $\lambda=\pm\lambda_{1}$, that is,
\begin{equation}\label{2-int-g1}
	\int_{-\infty}^{+\infty}g_{+,1}(\lambda;x,t)\dd x=0,\ \ \ \ {\rm when}\ \lambda\in\Sigma_{+}\ {\rm or}\ \lambda=\pm\lambda_{1}.
\end{equation}
	Similarly, we denote
	\begin{equation*}
		\int_{-\infty}^{+\infty}\Big(\big(\widehat{\phi}_{11}^{+}(\lambda;x,t)\big)^{2}q^{[1]*}_{x}(x,t){-}\big(\widehat{\phi}_{21}^{+}(\lambda;x,t)\big)^{2}q^{[1]}_{x}(x,t)\Big)\dd x:=\int_{-\infty}^{+\infty}g_{-,1}(\lambda;x,t)\dd x.
	\end{equation*}
By using the same method, we can obtain
	\begin{equation}\label{2-int-g2}
		\int_{-\infty}^{+\infty}g_{-,1}(\lambda;x,t)\dd x=0,\ \ \ \ {\rm when}\ \lambda\in\Sigma_{-}\ {\rm or}\ \lambda=\pm\lambda_{1}^{*}.
	\end{equation}
	
	The expressions \eqref{2-int-g1} and \eqref{2-int-g2} provide insight into the decomposition of the integral on $\Gamma_{+}$ or $\Gamma_{-}$ within the right-hand side of equation \eqref{2-exact-fFL}. This decomposition allows us to separate the integral into continuous and discrete components. It is worth noting that the discrete portions represent the residues of the integrand functions in \eqref{2-exact-fFL}, which are exactly at the zeros of $s_{11}(\lambda;t)$ and $s_{22}(\lambda;t)$. We denote
	\begin{equation*}
		|2\lambda^{2}|^{2\epsilon}\partial_{x}\left(\left(\widehat{\phi}_{11}^{-}(\lambda;x,t)\right)^2\right):=g_{+,2}(\lambda;x,t),\ \ \ \ |2\lambda^{2}|^{2\epsilon}\partial_{x}\left(\left(\widehat{\phi}_{12}^{-}(\lambda;x,t)\right)^2\right):=g_{-,2}(\lambda;x,t).
	\end{equation*}
	Then we can rewrite the right part in \eqref{2-exact-fFL} as follows by combining the equations \eqref{2-pf-s11+s22}, \eqref{2-int-g1} and \eqref{2-int-g2},
	\begin{multline}\label{2-Gamma+}
		\int_{\Gamma_{+}}\frac{g_{+,2}(\lambda;x,t)}{\lambda^{3}s_{11}^2(\lambda;t)}\int_{-\infty}^{+\infty}g_{+,1}(\lambda;y,t)\dd y\dd\lambda=2\pi\ii\sum\limits_{j=1}^{2}\left(\frac{g_{+,2}(\lambda_{j};x,t)}{\lambda_{j}^{3}s^{'2}_{11}(\lambda_{j};t)}\int_{-\infty}^{+\infty}g_{+,1}'(\lambda_{j};y,t)\dd y\right)\\
		=\dfrac{4\pi\ii\xi\eta(2(\xi^2+\eta^2))^{2\epsilon}\nu_{21}^{*}}{(\xi+\ii\eta)^{4}(\xi-\ii\eta)\nu_{11}^{*}}\ee^{2\ii\theta_{1I}+\varpi_{1}}{\rm sech}^{3}(2\theta_{1R}+\varpi_{1})\left((\xi+\ii\eta)^{2}\ee^{-2\theta_{1R}-\varpi_{1}}{\rm cosh}(2\theta_{1R}+\varpi_{1})-4\ii\xi\eta\right),
	\end{multline}
\begin{multline}\label{2-Gamma-}
	\int_{\Gamma_{-}}\frac{g_{-,2}(\lambda;x,t)}{\lambda^{3}s_{22}^2(\lambda;t)}\int_{-\infty}^{+\infty}g_{-,1}(\lambda;y,t)\dd y\dd\lambda=-2\pi\ii\sum\limits_{j=1}^{2}\left(\frac{g_{-,2}(\lambda_{j}^{*};x,t)}{\lambda_{j}^{*3}s^{'2}_{22}(\lambda_{j}^{*};t)}\int_{-\infty}^{+\infty}g_{-,1}'(\lambda_{j}^{*};y,t)\dd y\right)\\
	=-\dfrac{4\pi\ii\xi\eta(2(\xi^2+\eta^2))^{2\epsilon}\nu_{21}^{*}}{(\xi+\ii\eta)^{2}(\xi-\ii\eta)^{3}\nu_{11}^{*}}\ee^{2\ii\theta_{1I}+\varpi_{1}}{\rm sech}^{3}(2\theta_{1R}+\varpi_{1})\left((\xi-\ii\eta)^{2}\ee^{2\theta_{1R}+\varpi_{1}}{\rm cosh}(2\theta_{1R}+\varpi_{1})+4\ii\xi\eta\right),
\end{multline}
	where the superscript $'$
	denotes the partial derivative with respect to $\lambda$. Moreover, we can directly obtain the derivative of $q^{[1]}(x,t)$ with respect to $x$ and $t$ by according to the solution \eqref{2-onesoliton},
	\begin{equation}\label{2-qxt}
		q^{[1]}_{xt}(x,t)=\dfrac{\ii\xi\eta(2(\xi^2+\eta^2))^{2\epsilon-3}\nu_{21}^{*}}{4(\xi+\ii\eta)\nu_{11}^{*}}\ee^{2\ii\theta_{1I}+\varpi_{1}}{\rm sech}^{3}(2\theta_{1R}+\varpi_{1})\left((\xi^2+\eta^2)^{2}{\rm cosh}^{2}(2\theta_{1R}+\varpi_{1})-8\xi^2\eta^2\right).
	\end{equation}
Substituting \eqref{2-Gamma+}, \eqref{2-Gamma-} and \eqref{2-qxt} into \eqref{2-exact-fFL}, then the left and right sides are equal. This completes the proof.
\end{proof}
Based on the solution $q^{[1]}(x,t)$ (i.e.,\eqref{2-onesoliton}), we can easily obtain 
\begin{equation*}
	|q^{[1]}(x,t)|^{2}=\dfrac{8\xi^2\eta^2}{(\xi^2+\eta^2)^{2}\left((\xi^2+\eta^2){\rm cosh}\left(4\theta_{1R}+2\ln\left|\frac{\nu_{11}}{\nu_{21}}\right|+(\xi^2-\eta^2)\right)\right)}.
\end{equation*}
The maximum value of $|q^{[1]}(x,t)|$ can be precisely calculated as $\frac{2\eta}{\xi^2+\eta^2}$. Additionally, it is important to note that the velocity of the solitary wave is $v_{w}^{[1]}=-\left(2(\xi^2+\eta^2)\right)^{2\epsilon-2}$. By selecting appropriate parameters, we depict the temporal evolution of wave propagation and investigated the impact of the small parameter $\epsilon$ on this phenomenon, as illustrated in Fig.\ref{fig:wave-direction}. A comprehensive analysis of the figure reveals that the soliton exhibits a characteristic leftward movement, indicative of a left-going traveling-wave soliton. Interestingly, the wave moves farther as $\epsilon$ increases, a phenomenon that is corroborated through the observation of the wave velocity (see Fig.\ref{fig:velocity}). Moreover, the peak amplitude of the soliton achieves $1.6$, based on the chosen parameters in Fig.\ref{fig:wave-direction}.
\begin{figure}[h]
	\centering
	\includegraphics[width=0.9\linewidth]{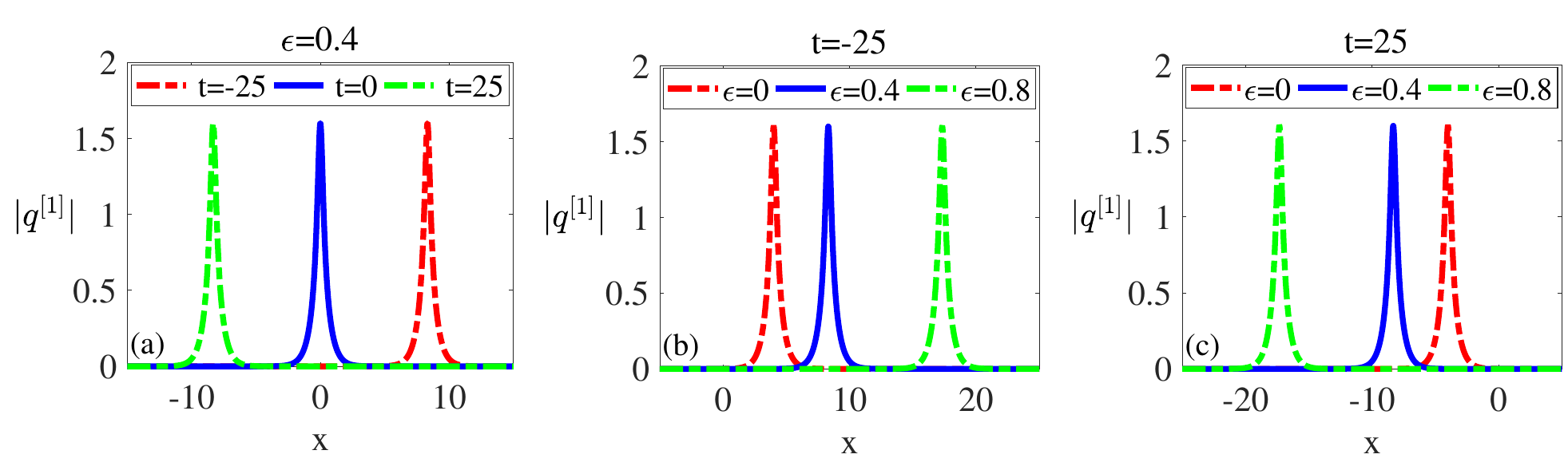}
	\caption{The profiles of the fractional one-soliton solution $q^{[1]}$. (a) the profiles of the fractional one-soliton at $\epsilon=0.4$. It clearly shows that as time increases, the soliton moves in the negative direction along the $x$-axis. (b)-(c) the profiles of the fractional one-soliton at fixed values of $t$. By examining these two figures together, it becomes apparent that as $\epsilon$ increases, the soliton moves farther. The parameters are $\xi=0.5,\ \eta=1,\ \nu_{11}=\nu_{21}=1.$}
	\label{fig:wave-direction}
\end{figure}
\begin{figure}[h]
	\centering
	\includegraphics[width=0.65\linewidth]{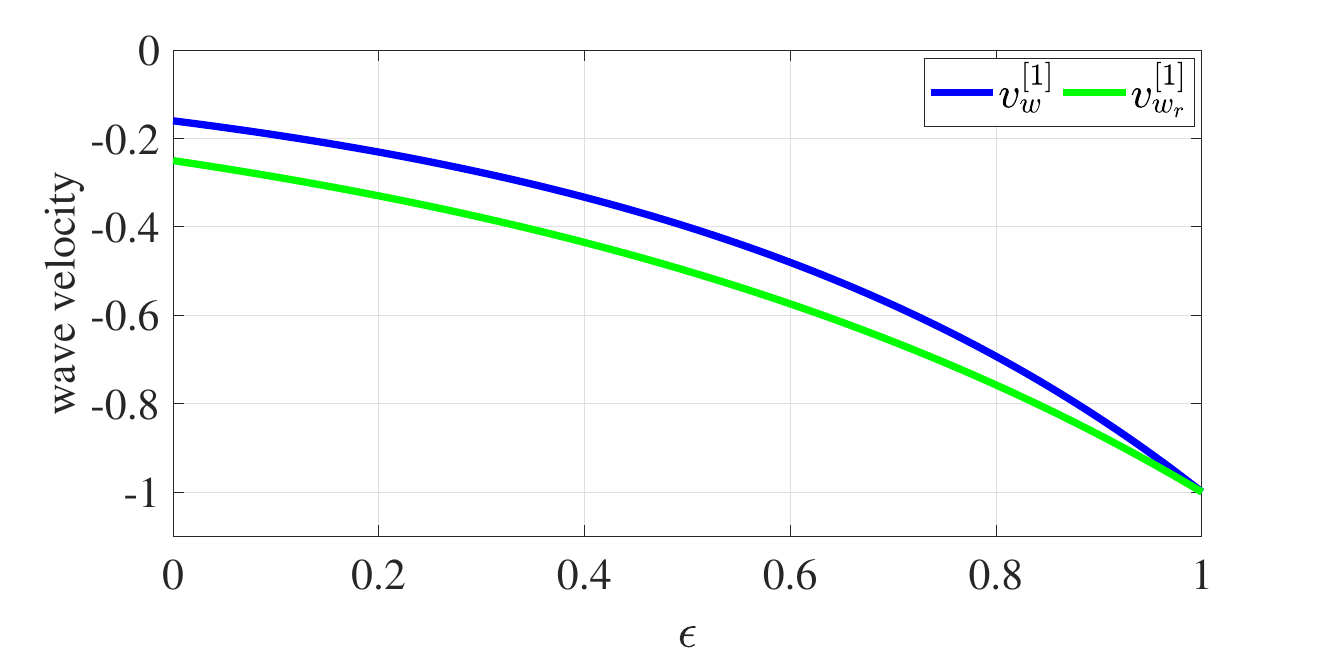}
	\caption{The solid blue curve ($\xi=0.5$): the wave velocity $v_{w}^{[1]}$ corresponds to the fractional one-soliton solution $q^{[1]}$. The solid green curve ($\xi=0$): the wave velocity $v_{w_{r}}^{[1]}$ corresponds to the fractional rational solution $q_{r}^{[1]}$. The negative values of the wave velocities suggest that the solitons are all left-going traveling-wave solitons. And the absolute values of wave velocities become larger as $\epsilon$ increases, implying that the solitons move farther as $\epsilon$ increases. The parameters are fixed by $\eta=1,\ \nu_{11}=\nu_{21}=1$.}
	\label{fig:velocity}
\end{figure}

Furthermore, the rational solution of the fFL equation can be explored by approaching the limit of the fractional one-soliton solution $q^{[1]}(x,t)$. Through the thorough analysis and calculations, it has been determined that the rational solution $q^{[1]}_{r}(x,t)$ emerges exclusively only when $\xi\to 0$ and satisfies the constraint $|\nu_{11}|=|\nu_{21}|$. The detailed expression of this rational solution is presented below:
\begin{equation*}
	q^{[1]}_{r}(x,t)=\dfrac{2\ii\eta \nu_{21}^{*}\exp\left(2\ii\eta^2 x-\ii (2\eta^2)^{2\epsilon-1}t\right)}{\nu_{11}^{*}\left(4\ii\eta^4 x+\ii(2\eta^2)^{2\epsilon}t+\eta^2\right)}.
\end{equation*}

By considering the fractional rational solution $q^{[1]}_{r}(x,t)$ and selecting parameters consistent with Fig.\ref{fig:wave-direction}, we can ascertain the temporal evolution of wave propagation and effect of $\epsilon$ on this phenomenon, as illustrated in Figure \ref{fig:wave-direction-rationalsolution}. Evidently, these characteristics align with the properties exhibited by the fractional one-soliton solution.  It is essential to note that since the fractional rational solution is obtained in the limit $\xi\to 0$, the wave peak naturally attains a value of $2$.
\begin{figure}[ht]
	\centering
	\includegraphics[width=0.9\linewidth]{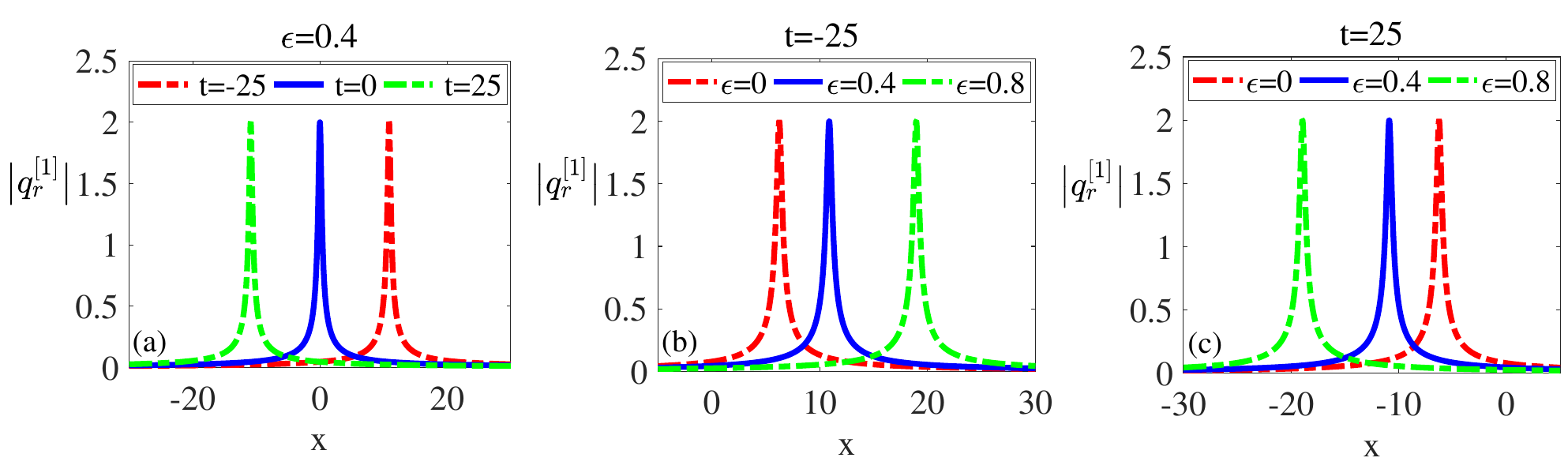}
	\caption{The profiles of the fractional resonant soliton corresponding to the fractional rational solution $q^{[1]}_{r}$ when the time $t$ or the small parameter $\epsilon$ is fixed. The dynamical behavior exhibited by the fractional resonant soliton during its motion is identical to that of the fractional one-soliton (see Fig.\ref{fig:wave-direction}). The parameters are given by $\eta=1,\ \nu_{11}=\nu_{21}=1.$}
	\label{fig:wave-direction-rationalsolution}
\end{figure}
\subsection{Fractional $N$-soliton solution}
Next, we will analyze the asymptotic states of the fractional $N$-soliton \eqref{2-recover-qx} as $|t|\to\infty$. Here we need use the well-known Cauchy determinant:
\begin{equation*}
	\det\left(\left[\dfrac{1}{x_{i}+y_{j}}\right]_{i,j=1}^{n}\right)=\dfrac{\prod\limits_{1\leqslant i<j\leqslant n}(x_{i}-x_{j})(y_{i}-y_{j})}{\prod\limits_{i,j=1}^{n}(x_{i}+y_{j})}.
\end{equation*}
\begin{prop}\label{prop-Nsoliton}
	For the fractional Fokas--Lenells equation, its fractional $N$-soliton solution can be approximated as the sum of $N$ fractional single-soliton solutions, as $|t|$ tends to $\infty$.
\end{prop}
\begin{proof}
We assume $\lambda_{j}=\xi_{j}+\ii\eta_{j}$, which belongs to the ${\rm\uppercase\expandafter{\romannumeral1}}$ quadrant, and define $\theta_{j}=\theta_{jR}+\ii\theta_{jI}$. Without loss of
generality, we suppose $v_{1}<v_{2}<\cdots<v_{N}<0$, where the velocity of the single soliton is $v_{j}=-\left(2(\xi_{j}^2+\eta_{j}^2)\right)^{2\epsilon-2}$. Firstly, we consider the case of $t\to-\infty$. 
In the reference frame that is moving with a velocity of $v_{k}\ (1\leqslant k\leqslant N)$, there exists
\begin{equation*}
\theta_{kR}=2\xi_{k}\eta_{k}\left(x+\left(2(\xi_{k}^2+\eta_{k}^2)\right)^{2\epsilon-2}t\right)=\mathcal{O}(1),
\end{equation*}
then
\begin{equation*}
	\theta_{jR}=2\xi_{j}\eta_{j}\left(x-v_{k}t\right)+2\xi_{j}\eta_{j}(v_{k}-v_{j})t\to\begin{cases}
	-\infty,\ \ \ \ 	j=1,2,\cdots,k-1,\\[2pt]
		+\infty,\ \ \ \ j=k+1,k+2,\cdots,N,
	\end{cases}.
\end{equation*}
So $|\textbf{v}_{j}\rangle=\begin{bmatrix}
	0,&1
\end{bmatrix}^{\top}$ when $j=1,2,\cdots,k-1$, and  $|\textbf{v}_{j}\rangle=\begin{bmatrix}
1,&0
\end{bmatrix}^{\top}$ when $j=k+1,k+2,\cdots,N$. By combining the solution \eqref{2-recover-qx}, we have
\begin{equation*}
	q^{[N],-}_{k}(x,t)=\frac{\det\widehat{\bf P}^{-}_{k}}{\det{\bf P}^{-}_{k}},
\end{equation*}
where
\begin{equation}\label{2-P-}
	\begin{split}
	&\det{\bf P}^{-}_{k}=\left|\begin{matrix}
		\frac{1}{\lambda_{1}^{2}-\lambda_{1}^{*2}}&\cdots&\frac{1}{\lambda_{1}^{2}-\lambda_{k-1}^{*2}}&\frac{\lambda_{k}^{*}\nu_{2k}^{*}\ee^{-\theta_{k}^{*}}}{\lambda_{1}^{2}-\lambda_{k}^{*2}}&0&\cdots&0\\
		\vdots&\ddots&\vdots&\vdots&\vdots&\ddots&\vdots\\
		\frac{1}{\lambda_{k-1}^{2}-\lambda_{1}^{*2}}&\cdots&\frac{1}{\lambda_{k-1}^{2}-\lambda_{k-1}^{*2}}&\frac{\lambda_{k}^{*}\nu_{2k}^{*}\ee^{-\theta_{k}^{*}}}{\lambda_{k-1}^{2}-\lambda_{k}^{*2}}&0&\cdots&0\\[8pt]
		\frac{\nu_{2k}\ee^{-\theta_{k}}}{\lambda_{k}^{2}-\lambda_{1}^{*2}}&\cdots&\frac{\nu_{2k}\ee^{-\theta_{k}}}{\lambda_{k}^{2}-\lambda_{k-1}^{*2}}&\frac{\lambda_{k}|\nu_{1k}|^2\ee^{2\theta_{kR}}+\lambda_{k}^{*}|\nu_{2k}|^2\ee^{-2\theta_{kR}}}{\lambda_{k}^{2}-\lambda_{k}^{*2}}&\frac{\lambda_{k}\nu_{1k}\ee^{\theta_{k}}}{\lambda_{k}^{2}-\lambda_{k+1}^{*2}}&\cdots&\frac{\lambda_{k}\nu_{1k}\ee^{\theta_{k}}}{\lambda_{k}^{2}-\lambda_{N}^{*2}}\\[8pt]
		0&\cdots&0&\frac{\nu_{1k}^{*}\ee^{\theta_{k}^{*}}}{\lambda_{k+1}^{2}-\lambda_{k}^{*2}}&\frac{1}{\lambda_{k+1}^{2}-\lambda_{k+1}^{*2}}&\cdots&\frac{1}{\lambda_{k+1}^{2}-\lambda_{N}^{*2}}\\
		\vdots&\ddots&\vdots&\vdots&\vdots&\ddots&\vdots\\
		0&\cdots&0&\frac{\nu_{1k}^{*}\ee^{\theta_{k}^{*}}}{\lambda_{N}^{2}-\lambda_{k}^{*2}}&\frac{1}{\lambda_{N}^{2}-\lambda_{k+1}^{*2}}&\cdots&\frac{1}{\lambda_{N}^{2}-\lambda_{N}^{*2}}
	\end{matrix}\right|,\\[3pt]
&\det{\bf \widehat{P}}^{-}_{k}=\left|\begin{matrix}
	{\bf P}^{-}_{k}&\hat{{\bf p}}_{1k}^{-}\\[5pt]
	\hat{{\bf p}}_{2k}^{-}&0
\end{matrix}\right|,\ \ \ 
\hat{{\bf p}}_{1k}^{-}=[\ \underbrace{0,\cdots,0}_{k-1},\ \nu_{1k}\ee^{\theta_{k}},\ \underbrace{1,\cdots,1}_{N-k}\ ]^{\top},\ \ \ 
\hat{{\bf p}}_{2k}^{-}=[\ \underbrace{1,\cdots,1}_{k-1},\ \nu_{2k}^{*}\ee^{-\theta_{k}^{*}},\ \underbrace{0,\cdots,0}_{N-k}\ ].
\end{split}
\end{equation}
For convenience, we introduce some notations:
\begin{equation*}
	\begin{split}
		&C_{m,n}=\det\left(\left[\frac{1}{\lambda_{i}^{2}-\lambda_{j}^{*2}}\right]_{i,j=m}^{n}\right)=\dfrac{\prod\limits_{m\leqslant i<j\leqslant n}\left(-\left|\lambda_{i}^{2}-\lambda_{j}^{2}\right|^{2}\right)}{\prod\limits_{i,j=m}^{n}(\lambda_{i}^{2}-\lambda^{*2}_{j})},\ \ \ \ 1\leqslant m<n\leqslant N,\\[3pt]
		&C_{m,m}=\dfrac{1}{\lambda_{m}^{2}-\lambda_{m}^{*2}},\ \ \ \ C_{1,0}=C_{N+1,N}=1,
	\end{split}
\end{equation*}
\begin{equation*}
	\begin{split}
		&C^{-}_{1,i}=\left|\begin{matrix}
			\frac{1}{\lambda_{1}^{2}-\lambda_{1}^{*2}}&\cdots&\frac{1}{\lambda_{1}^{2}-\lambda_{i}^{*2}}\\[2pt]
			\vdots&\ddots&\vdots\\[2pt]
			\frac{1}{\lambda_{i-1}^{2}-\lambda_{1}^{*2}}&\cdots&\frac{1}{\lambda_{i-1}^{2}-\lambda_{i}^{*2}}\\[8pt]
			\frac{1}{\lambda_{1}^{*2}}&\cdots&\frac{1}{\lambda_{i}^{*2}}
		\end{matrix}\right|=\prod\limits_{1\leqslant m<n<i}\left(-\left|\lambda_{m}^{2}-\lambda_{n}^{2}\right|^{2}\right)\prod\limits_{m=1}^{i-1}\prod\limits_{n=1}^{i}\frac{\lambda_{m}^{2}(\lambda_{m}^{*2}-\lambda_{i}^{*2})}{\lambda_{n}^{*2}(\lambda_{m}^{2}-\lambda_{n}^{*2})},\\[4pt]
	&C_{1,i}^{+}=\left|\begin{matrix}
		\frac{1}{\lambda_{1}^{2}-\lambda_{1}^{*2}}&\cdots&\frac{1}{\lambda_{1}^{2}-\lambda_{i-1}^{*2}}&\frac{1}{\lambda_{1}^2}\\[4pt]
		\vdots&\ddots&\vdots&\vdots\\[4pt]
		\frac{1}{\lambda_{i}^{2}-\lambda_{1}^{*2}}&\cdots&\frac{1}{\lambda_{i}^{2}-\lambda_{i-1}^{*2}}&\frac{1}{\lambda_{i}^2}
	\end{matrix}\right|=\prod\limits_{1\leqslant m<n<i}\left(-\left|\lambda_{m}^{2}-\lambda_{n}^{2}\right|^{2}\right)\prod\limits_{m=1}^{i-1}\prod\limits_{n=1}^{i}\frac{\lambda_{m}^{*2}(\lambda_{i}^{2}-\lambda_{m}^{2})}{\lambda_{n}^{2}(\lambda_{n}^{2}-\lambda_{m}^{*2})},\\[4pt]
&\bar{C}^{-}_{j,N}=\left|\begin{matrix}
	\frac{1}{\lambda_{j}^2}&\frac{1}{\lambda_{j}^{2}-\lambda_{j+1}^{*2}}&\cdots&\frac{1}{\lambda_{j}^{2}-\lambda_{N}^{*2}}\\[4pt]
	\vdots&\vdots&\ddots&\vdots\\[4pt]
	\frac{1}{\lambda_{N}^2}&\frac{1}{\lambda_{N}^{2}-\lambda_{j+1}^{*2}}&\cdots&\frac{1}{\lambda_{N}^{2}-\lambda_{N}^{*2}}
\end{matrix}\right|=\prod\limits_{j< m<n\leqslant N}\left(-\left|\lambda_{m}^{2}-\lambda_{n}^{2}\right|^{2}\right)\prod\limits_{m=j+1}^{N}\prod\limits_{n=j}^{N}\frac{\lambda_{m}^{*2}(\lambda_{j}^{2}-\lambda_{m}^{2})}{\lambda_{n}^{2}(\lambda_{n}^{2}-\lambda_{m}^{*2})},\\[4pt]
&\bar{C}_{j,N}^{+}=\left|\begin{matrix}
	\frac{1}{\lambda_{j}^{*2}}&\cdots&\frac{1}{\lambda_{N}^{*2}}\\[6pt]
	\frac{1}{\lambda_{j+1}^{2}-\lambda_{j}^{*2}}&\cdots&\frac{1}{\lambda_{j+1}^{2}-\lambda_{N}^{*2}}\\[2pt]
	\vdots&\ddots&\vdots\\[2pt]
	\frac{1}{\lambda_{N}^{2}-\lambda_{j}^{*2}}&\cdots&\frac{1}{\lambda_{N}^{2}-\lambda_{N}^{*2}}
\end{matrix}\right|=\prod\limits_{j< m<n\leqslant N}\left(-\left|\lambda_{m}^{2}-\lambda_{n}^{2}\right|^{2}\right)\prod\limits_{m=j+1}^{N}\prod\limits_{n=j}^{N}\frac{\lambda_{m}^{2}(\lambda_{m}^{*2}-\lambda_{j}^{*2})}{\lambda_{n}^{*2}(\lambda_{m}^{*2}-\lambda_{n}^{2})},\\[4pt]
&C_{1,1}^{-}=\frac{1}{\lambda_{1}^{*2}},\ \ \ \ C_{1,1}^{+}=\frac{1}{\lambda_{1}^{2}},\ \ \ \ \bar{C}_{N,N}^{-}=\frac{1}{\lambda_{N}^{2}},\ \ \ \ \bar{C}_{N,N}^{+}=\frac{1}{\lambda_{N}^{*2}},
	\end{split}
\end{equation*}
where $1<i\leqslant N,\ 1\leqslant j<N$. 
By direct calculation, we can rewrite \eqref{2-P-} as follows:
\begin{equation*}
	\begin{split}
		\det{\bf P}^{-}_{k}=&\left|\begin{matrix}\begin{array}{cccc:ccc}
				\frac{1}{\lambda_{1}^{2}-\lambda_{1}^{*2}}&\cdots&\frac{1}{\lambda_{1}^{2}-\lambda_{k-1}^{*2}}&\frac{\lambda_{k}^{*}\nu_{2k}^{*}\ee^{-\theta_{k}^{*}}}{\lambda_{1}^{2}-\lambda_{k}^{*2}}&0&\cdots&0\\
				\vdots&\ddots&\vdots&\vdots&\vdots&\ddots&\vdots\\
				\frac{1}{\lambda_{k-1}^{2}-\lambda_{1}^{*2}}&\cdots&\frac{1}{\lambda_{k-1}^{2}-\lambda_{k-1}^{*2}}&\frac{\lambda_{k}^{*}\nu_{2k}^{*}\ee^{-\theta_{k}^{*}}}{\lambda_{k-1}^{2}-\lambda_{k}^{*2}}&0&\cdots&0\\[8pt]
				\frac{\nu_{2k}\ee^{-\theta_{k}}}{\lambda_{k}^{2}-\lambda_{1}^{*2}}&\cdots&\frac{\nu_{2k}\ee^{-\theta_{k}}}{\lambda_{k}^{2}-\lambda_{k-1}^{*2}}&\frac{\lambda_{k}^{*}|\nu_{2k}|^2\ee^{-2\theta_{kR}}}{\lambda_{k}^{2}-\lambda_{k}^{*2}}&\frac{\lambda_{k}\nu_{1k}\ee^{\theta_{k}}}{\lambda_{k}^{2}-\lambda_{k+1}^{*2}}&\cdots&\frac{\lambda_{k}\nu_{1k}\ee^{\theta_{k}}}{\lambda_{k}^{2}-\lambda_{N}^{*2}}\\[8pt]
				\hdashline\\[-6pt]
				0&\cdots&0&0&\frac{1}{\lambda_{k+1}^{2}-\lambda_{k+1}^{*2}}&\cdots&\frac{1}{\lambda_{k+1}^{2}-\lambda_{N}^{*2}}\\
				\vdots&\ddots&\vdots&\vdots&\vdots&\ddots&\vdots\\
				0&\cdots&0&0&\frac{1}{\lambda_{N}^{2}-\lambda_{k+1}^{*2}}&\cdots&\frac{1}{\lambda_{N}^{2}-\lambda_{N}^{*2}}
		\end{array}\end{matrix}\right|\\[3pt]
		+&\left|\begin{matrix}\begin{array}{ccc:cccc}
				\frac{1}{\lambda_{1}^{2}-\lambda_{1}^{*2}}&\cdots&\frac{1}{\lambda_{1}^{2}-\lambda_{k-1}^{*2}}&0&0&\cdots&0\\
				\vdots&\ddots&\vdots&\vdots&\vdots&\ddots&\vdots\\
				\frac{1}{\lambda_{k-1}^{2}-\lambda_{1}^{*2}}&\cdots&\frac{1}{\lambda_{k-1}^{2}-\lambda_{k-1}^{*2}}&0&0&\cdots&0\\[8pt]
				\hdashline\\[-6pt]
				\frac{\nu_{2k}\ee^{-\theta_{k}}}{\lambda_{k}^{2}-\lambda_{1}^{*2}}&\cdots&\frac{\nu_{2k}\ee^{-\theta_{k}}}{\lambda_{k}^{2}-\lambda_{k-1}^{*2}}&\frac{\lambda_{k}|\nu_{1k}|^2\ee^{2\theta_{kR}}}{\lambda_{k}^{2}-\lambda_{k}^{*2}}&\frac{\lambda_{k}\nu_{1k}\ee^{\theta_{k}}}{\lambda_{k}^{2}-\lambda_{k+1}^{*2}}&\cdots&\frac{\lambda_{k}\nu_{1k}\ee^{\theta_{k}}}{\lambda_{k}^{2}-\lambda_{N}^{*2}}\\[8pt]
				0&\cdots&0&\frac{\nu_{1k}^{*}\ee^{\theta_{k}^{*}}}{\lambda_{k+1}^{2}-\lambda_{k}^{*2}}&\frac{1}{\lambda_{k+1}^{2}-\lambda_{k+1}^{*2}}&\cdots&\frac{1}{\lambda_{k+1}^{2}-\lambda_{N}^{*2}}\\
				\vdots&\ddots&\vdots&\vdots&\vdots&\ddots&\vdots\\
				0&\cdots&0&\frac{\nu_{1k}^{*}\ee^{\theta_{k}^{*}}}{\lambda_{N}^{2}-\lambda_{k}^{*2}}&\frac{1}{\lambda_{N}^{2}-\lambda_{k+1}^{*2}}&\cdots&\frac{1}{\lambda_{N}^{2}-\lambda_{N}^{*2}}
		\end{array}\end{matrix}\right|\\[3pt]
	=&C_{1,k-1}C_{k,N}\lambda_{k}|\nu_{1k}|^{2}\ee^{2\theta_{kR}}+C_{1,k}C_{k+1,N}\lambda_{k}^{*}|\nu_{2k}|^{2}\ee^{-2\theta_{kR}},\\[4pt]
	\det{\bf \widehat{P}}^{-}_{k}=&
	C_{1,k}^{-}\bar{C}_{k,N}^{-}|\lambda_{k}|^{2}\nu_{1k}\nu_{2k}^{*}\ee^{2\ii\theta_{kI}}.
	\end{split}
\end{equation*}
Thus, along the trajectory $x-v_{k}t=const$, we have
\begin{equation*}
q_{k}^{[N],-}(x,t)=\frac{2\ii\xi_{k}\eta_{k}\nu_{2k}^{*}}{(\xi_{k}+\ii\eta_{k})^{2}(\xi_{k}-\ii\eta_{k})\nu_{1k}^{*}}\exp\left(2\ii\theta_{kI}+\varpi_{k}-c_{k}\right){\rm sech}\left(2\theta_{kR}+\varpi_{k}-\delta_{k}\right),
\end{equation*}
where
\begin{equation*}
	\begin{split}
	&\varpi_{k}=\ln\left(\frac{\lambda_{k}|\nu_{1k}|}{|\lambda_{k}||\nu_{2k}|}\right),\ \ \ \ \delta_{k}=\ln\left(\prod\limits_{m=1}^{k-1}\prod\limits_{n=k+1}^{N}\left|\dfrac{(\lambda_{m}^{2}-\lambda_{k}^{2})(\lambda_{n}^{2}-\lambda_{k}^{*2})}{(\lambda_{m}^{2}-\lambda_{k}^{*2})(\lambda_{n}^{2}-\lambda_{k}^{2})}\right|\right),\\ &c_{k}=\frac{1}{2}\ln\left(\prod\limits_{m=1}^{k-1}\prod\limits_{n=k+1}^{N}\dfrac{\tau_{m}^{*}\tau_{n}}{\tau_{m}\tau_{n}^{*}}\right),\ \ \ \ \tau_{j}=\lambda_{j}^{4}(\lambda_{k}^{2}-\lambda_{j}^{*2})(\lambda_{k}^{*2}-\lambda_{j}^{*2}),\ \ j\neq k,
		\end{split}
\end{equation*}
and we define $\prod\limits_{m=1}^{0}f_{m}=\prod\limits_{n=N+1}^{N}f_{m}=1$ with arbitrary smooth function $f_{m}$. 
Therefore, when $t\to-\infty$,
\begin{equation*}
	q^{[N],-}(x,t)=\sum\limits_{k=1}^{N}q^{[N],-}_{k}(x,t)+\mathcal{O}\left(\ee^{\iota t}\right),\ \ \ \ \iota=\max\limits_{i\neq j}\left(2\xi_{j}\eta_{j}\left|v_{i}-v_{j}\right|\right),\ \ i,j=1,\cdots,N.
\end{equation*}

Using the same method, we can analyze the asymptotic form of the fractional $N$-soliton solution when $t\to+\infty$, as shown below:
\begin{equation*}
	q^{[N],+}(x,t)=\sum\limits_{k=1}^{N}q_{k}^{[N],+}(x,t)+\mathcal{O}\left(\ee^{-\iota t}\right),
\end{equation*}
where
\begin{equation*}
	q_{k}^{[N],+}(x,t)=\frac{2\ii\xi_{k}\eta_{k}\nu_{2k}^{*}}{(\xi_{k}+\ii\eta_{k})^{2}(\xi_{k}-\ii\eta_{k})\nu_{1k}^{*}}\exp\left(2\ii\theta_{kI}+\varpi_{k}+c_{k}\right){\rm sech}\left(2\theta_{kR}+\varpi_{k}+\delta_{k}\right).
\end{equation*}
Consequently, we complete the proof.
\end{proof}
We take the fractional three-soliton solution as an example to compare the exact solution $q^{[3]}(x,t)$ and the asymptotic solutions $q^{[3],\pm}(x,t)$ at $t=\pm 40$, respectively. It is shown that the decomposition of three-soliton is consistent with the exact three-soliton solution (see Fig.\ref{fig:asy-3soliton}), which exhibits the feature of elastic collisions as the classical multi-soliton of FL equation.
\begin{figure}[h]
	\centering
	\subfloat[Choosing $\epsilon=0$. (a1) the fractional three-soliton solution. (a2)-(a3) the comparison between the exact solution and the asymptotic solution when $t=\pm40$, respectively.]{\includegraphics[width=0.9\linewidth]{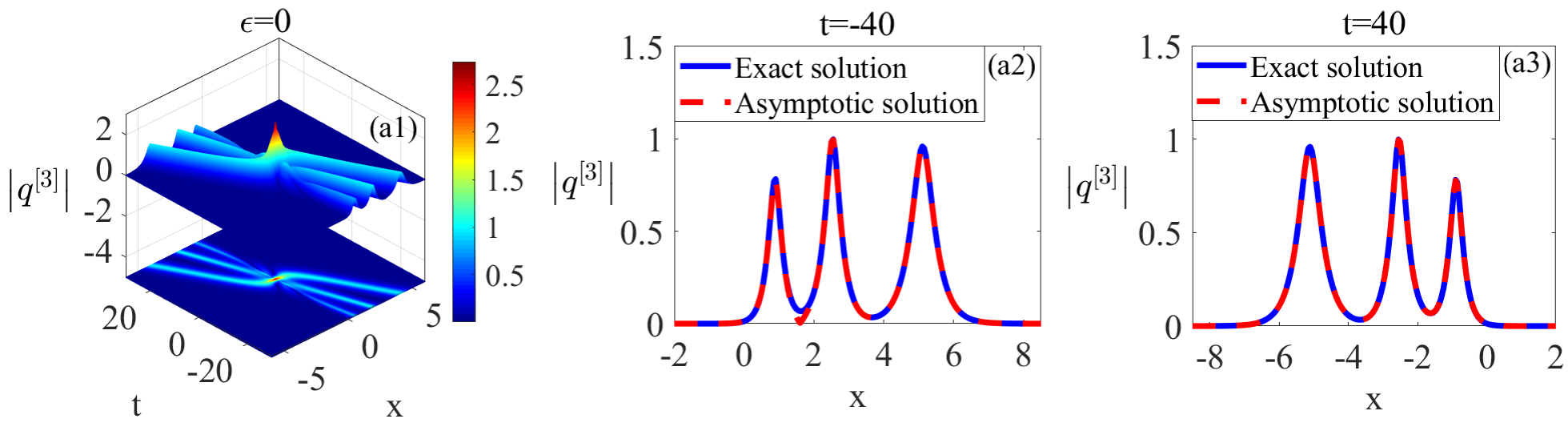}
		\label{fig:asy-3soliton-1}}\hfill
	\subfloat[Choosing $\epsilon=0.4$. (b1) the fractional three-soliton solution. (b2)-(b3) the comparison between the exact solution and the asymptotic solution when $t=\pm40$, respectively.]{\hspace{0.2cm}\includegraphics[width=0.9\linewidth]{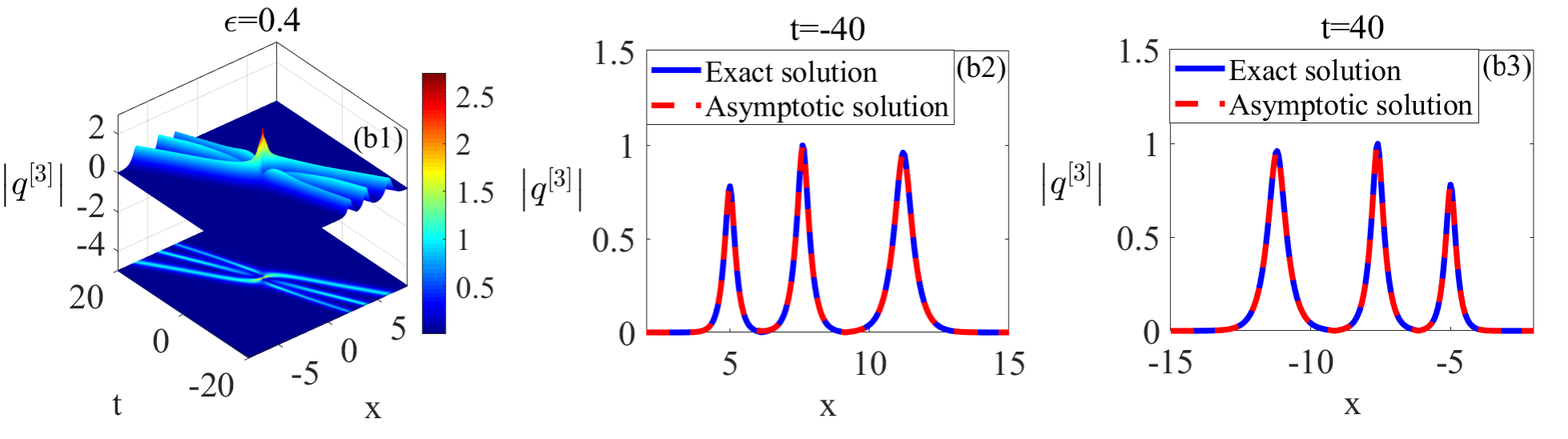}
		\label{fig:asy-3soliton-2}}\hfill
	\subfloat[Choosing $\epsilon=0.8$. (c1) the fractional three-soliton solution. (c2)-(c3) the comparison between the exact solution and the asymptotic solution when $t=\pm40$, respectively.]{\includegraphics[width=0.9\linewidth]{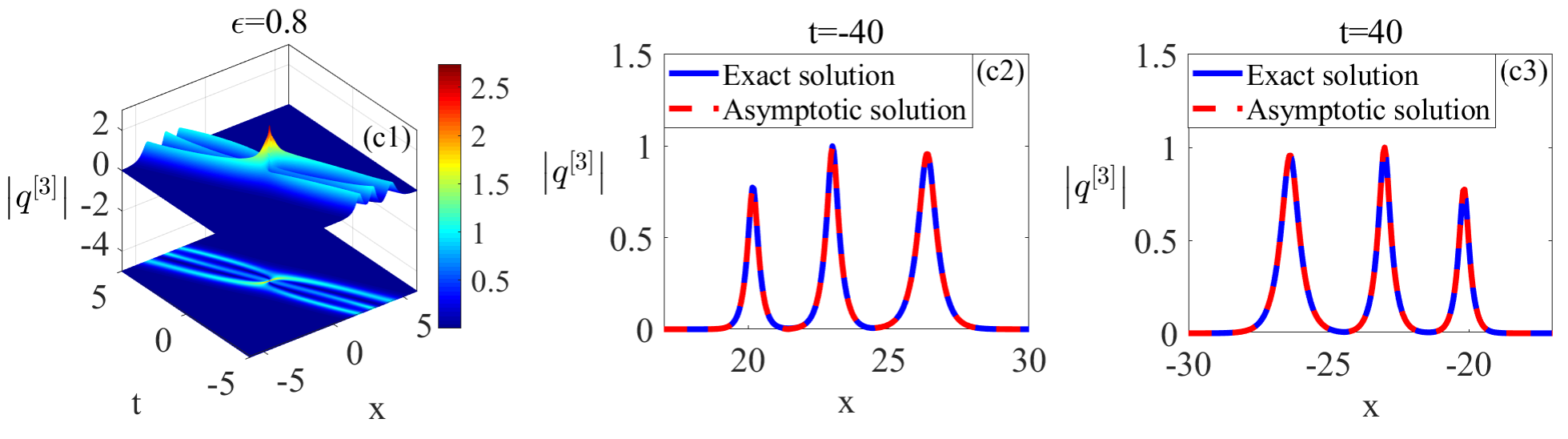}
		\label{fig:asy-3soliton-3}}
	\caption{Fractional three-soliton solution $q^{[3]}(x,t)$ and the comparison between the exact solution and the asymptotic solution. The figures (a)-(c) correspond to the cases of $\epsilon=0,0.4,0.8$, respectively. The solid blue curve and the dashed red curve are plotted by the exact solution and the asymptotic solution of the fractional three-soliton solution based on the proof process of the proposition \ref{prop-Nsoliton}, respectively. The parameters are taken by $\nu_{11}=\nu_{12}=\nu_{13}=\nu_{21}=\nu_{22}=\nu_{23}=1,\ \xi_{1}=1,\ \eta_{1}=\frac{3}{4},\ \xi_{2}=1,\ \eta_{2}=1,\ \xi_{3}=\frac{5}{4},\ \eta_{3}=1.$}
	\label{fig:asy-3soliton}
\end{figure}

\section{Conclusion}
In this paper, we present an integrable fractional form of the FL equation based on the fractional integrable system proposed by Ablowitz, Been and Carr in \cite{ablowitz2022fractional}. By constructing a suitable Riemann-Hilbert problem, we provide a determinant representation for the fractional $N$-soliton solution of the fFL equation in the case of reflectionless potential. Additionally, we offer a comprehensive proof of the fractional one-soliton solution and discuss its limit form, leading to the fractional rational solution of the fFL equation. We also analyze the dynamic behavior of solitons and observe that they all propagate as left-traveling waves. Furthermore, we find that as the parameter $\epsilon$ increases, the solitons exhibit longer travel distances. Finally, based on the determinant representation of the fractional $N$-soliton solution, we prove that the fractional $N$-soliton solution can be regarded as a linear superposition of $N$ fractional single-soliton solutions when $|t|\to\infty$.

\section*{Acknowledgements}
Liming Ling is supported by the National Natural Science Foundation of China (No.$12122105$).

\section*{Conflict of interest}
The authors declare that they have no conflict of interest with other people or organizations that may inappropriately influence the author's work.

\bibliographystyle{unsrt}

\bibliography{Reference-fFL}

\end{document}